\documentclass[lettersize,journal]{IEEEtran}
\usepackage{amsmath,amsfonts,amssymb,amsthm}
\usepackage{algorithmic}
\usepackage{algorithm}
\usepackage{array}
\usepackage[caption=false,font=normalsize,labelfont=sf,textfont=sf]{subfig}
\usepackage{mathrsfs}
\usepackage{mathtools}
\mathtoolsset{showonlyrefs,showmanualtags}
\usepackage{textcomp}
\usepackage{stfloats}
\usepackage{url}
\usepackage{verbatim}
\usepackage{graphicx}
\usepackage{cite}
\usepackage[colorlinks=true, linkcolor=blue, citecolor=red]{hyperref}
\usepackage{mathtools}
\DeclarePairedDelimiterX{\norm}[1]{\lVert}{\rVert}{#1}
\theoremstyle{plain}
\newtheorem{theorem}{Theorem}
\newtheorem{lemma}{Lemma} 
\newtheorem{claim}{Claim}
\newtheorem{corollary}{Corollary} 
\theoremstyle{definition}

\newtheorem{assumption}{Assumption}
\theoremstyle{remark}
\newtheorem{remark}{Remark}
\begin{document}
\title{Observer-Based Robust Control for an Aerial Manipulator System under Unknown
External Disturbances}
\author{Mayank Pandey, Sneha Gajbhiye~\IEEEmembership{}
\thanks{M Pandey,  and S Gajbhiye are with the Department of Electrical Engineering, Indian Institute of Technology Palakkad, Kerala, India-678623.
\texttt{mayankpandey012005@gmail.com, snehagajbhiye@iitpkd.ac.in}}%
}
\maketitle
\begin{abstract}
This paper addresses the mathematical modeling and control of an aerial manipulator system comprising a quadrotor as the uncrewed aerial vehicle and a robotic arm as the manipulator. The dynamic model is established by identifying the overall center-of-mass velocity and the system's orientation as constraints, yielding a simplified, two-decoupled subsystems: a locked (overall translation) subsystem and a shape-space (actuation or overall rotation) subsystem, both subjected to external disturbances. 
Since the system is mechanically coupled, a critical challenge arises where constant bounded disturbances in the first subsystem manifest as time-varying, state-dependent disturbances in the second subsystem. Given that the quadrotor is inherently unstable, the movement of the robotic manipulator (RM) during flight can further jeopardize the stability of the entire QRM system if these disturbances and coupling effects are not effectively managed. To address this, we present a continuous nonlinear disturbance observer-based feedback control law, which enables the independent control of each subsystem while systematically eliminating cross-coupling effects. The efficacy of the proposed controller is validated through multiple simulations emulating practical operating conditions, thereby substantiating its real-world applicability and highlighting the core contributions of this work.
\end{abstract}
\section{Introduction}\label{sec:introduction}
In the field of robotics and automation, an ideal robot is one that can perform multiple actions at a time, especially one that has a larger workspace. With this idea, drones, land-based locomotion vehicles, manipulator robots, and underwater vehicles are performing exceptionally well, but not sufficiently. For instance, land-based vehicles offer nearly infinite horizontal mobility but lack manipulation capabilities, whereas robotic arms provide high dexterity within a very confined reach. To bridge these gaps, researchers are increasingly focusing on the coupling of robotic systems, integrating mobile bases with manipulators to create hybrid platforms that combine high mobility with functional interaction. One such platform is a combination of an unmanned aerial vehicle (UAV) such as \textit{quadrotor}, and a manipulator system. The application of these aerial manipulator systems is leading to advancements in areas such as accessing post-earthquake zones, areas with chemical  or radioactive contamination \cite{shukla2016application,barchyn2017uav}, flood-affected regions to perform tasks like sample collection or delivering aid\cite{6094871,popek2018autonomous},  pick and place operations in industries \cite{9468935}, load transportation \cite{michael2011cooperative,villa2020survey,pounds2012stability,bernard2010load,bernard2011autonomous}, inspecting and potentially performing minor repairs on wind turbine blades at high altitudes, military and agriculture \cite{charron2020deleaves,alsalam2017autonomous}, performing non-destructive tests by physically contacting surfaces with sensors to check material integrity and detect defects or cracks in construction bridges or dams \cite{jimenez2017aerial,freeman2021aerial}, inspecting and fixing high-voltage electric lines without needing to shut down power or put human workers at risk \cite{app11136220,ollero2024application,ollero2025multi,machines11111024},  or interacting with articulated objects like opening doors or operating emergency switches, valve turning \cite{kim2015operating,orsag2014valve}, and placing and retrieving sensors in rugged or dense terrain for monitoring species or fire prevention \cite{10.1007/978-3-031-71360-6_28,everaerts2008use,tomic2012toward}. 
Initial methodologies in aerial manipulation sought to mitigate aerodynamic disturbances by avoiding sustained contact forces. For example, \cite{lindsey2012construction} utilized passive magnetic components to autonomously assemble truss-like 2.5-D structures without active force regulation. However, achieving true operational autonomy required transitioning from such contact-averse strategies to direct physical manipulation. Consequently, the integrated framework presented in \cite{6094871} introduced active aerial grasping and transport mechanisms, enabling quadrotors with lightweight grippers to handle objects with unknown inertial properties despite the inherent coupling forces. Recognizing that rigid grasping alters the vehicle's center of mass, subsequent research turned to flexible links. The authors in \cite{sreenath2013geometric} modeled a quadrotor with a cable-suspended load on $SE(3)\times \mathcal{S}^2$, establishing that this underactuated system is differentially flat with respect to the load position and quadrotor yaw, and proposed a geometric controller to track load positions. To overcome the swinging vulnerabilities of cable-suspended loads, research naturally transitioned from flexible strings to articulated multi-link structures. Two years later, \cite{ruggiero2015multilayer} introduced a multilayer control architecture combining mechanical counterbalancing, active disturbance compensation, and force estimation to mitigate the severe dynamic coupling induced by a moving, multi-degree-of-freedom robotic arm.\\
In parallel to geometric and classical PID approaches, several studies have explored model-based control and cascaded PID-based control strategies for aerial manipulation systems \cite{8809210, article_859848}owing to their simplicity, ease of implementation, and robustness to modeling uncertainties.
Backstepping-based nonlinear control has been widely adopted for aerial manipulators using full coupled dynamics, including coordinate-based and coordinate-free formulations for underactuated systems \cite{kobilarov2014nonlinear}. Experimental validation of such approaches has also been demonstrated on outdoor multirotor platforms carrying high-DOF manipulators, with improved performance compared to PID-based control \cite{heredia2014control}.
A hybrid force/motion control framework for quadrotors equipped with rigid or lightweight tools is proposed in \cite{6696849}, where the dynamics are transformed to the tool-tip space and passively decomposed into tangential and normal components for contact control, with internal stability conditions and bounded angular-rate behavior analytically established. In the presence of model uncertainties, \cite{7139851} introduced a CAD-driven design with an adaptive backstepping controller for robust trajectory tracking. 
Early aerial transportation research typically assumed simplified slung-load or rigid-body payload dynamics. However, modern aerial manipulation demands active, highly articulated multi-degree-of-freedom tracking. Modeling these complex systems requires transitioning from basic particle dynamics to an n-DOF Euler-Lagrangian framework that captures all coupled state-space interactions. A foundational model of this type was developed by Lippiello and Ruggiero \cite{LIPPIELLO2012704}, who formulated a generalized n-DOF coupled dynamic model stabilized via Cartesian impedance control. Despite its foundational rigor, this impedance framework struggles with tracking accuracy under sustained external disturbances, leading to notable steady-state offsets. Furthermore, the explicit formulation of fully coupled, multi-body kinematics yields a massive computational burden, which frequently causes control command delays and numerical instability during numerical simulation and real-time execution. \\
To circumvent the mathematical complexity of such highly coupled formulations, recent paradigms favor reduction of system dimensionality or structural decoupling, in which the UAV and manipulator dynamics are treated independently during control synthesis or by using the differential flatness property of the system. For instant, the author in \cite{author2021title}, reduced the complex quadrotor-arm dynamics to a lower-dimensional structured system to exploit its differential flatness and control it via optimization-based integrators. These optimization-driven trajectories entail heavy onboard computational costs and remain acutely sensitive to unmodeled external forces. Leveraging the intrinsic structure of multirotor dynamics, the authors in \cite{7759712} decompose the system into a translational subsystem governing the CoM motion and a rotational subsystem governing attitude by using differential flatness. To bridge the gap between structural decoupling and robust disturbance rejection without compounding computational overhead, energy-based paradigms have emerged as a powerful alternative, as it guaranties stable energy exchange with the environment and inherently promotes compliant and safe contact behavior, even in the presence of model uncertainties and external disturbances. The work in \cite{6213082,5718145,5638155} utilizes the passive decomposition method to decouple the system's kinetic energy into two subsystems locked and shape by projecting the dynamics onto the tangent and normal spaces of the configuration manifold. Similarly, \cite{6907674} applies this passive decomposition to an aerial manipulator to reveal a decoupled structure and ensures stability through a backstepping-like end-effector tracking control law, though it assumes a disturbance-free environment.  To address this limitation, \cite{https://doi.org/10.1002/rnc.5041}  demonstrates finite-time stability for decoupled dynamics in the presence of unknown external disturbances using sliding mode control. However, a significant drawback of this technique remains its discontinuous nature, which introduces control chattering and results in unwanted jerk within the system.  Building on these advancements, \cite{7875412} proposed a 2-DOF arm configuration capable of achieving full 6-DOF end-effector motion, supported by a disturbance observer-based robust control framework. Compliant physical interaction in aerial manipulators is addressed in \cite{liu2024coordinated} through a coordinated quadrotor–manipulator framework that combines force estimation without force sensors and an adaptive-weight MPC planner, enabling autonomous stiffness adjustment while guaranteeing passivity via a Lyapunov-like energy function and validated through human–robot interaction experiments. 
To ensure analytical stability under aggressive aerodynamic and dynamic disturbances, Chen \cite{1372532, CHEN2001329} proposed disturbance observer-based control strategies for nonlinear systems. By extending this framework to handle harmonic variations, these strategies demonstrate significantly enhanced robustness against model uncertainties and external perturbations.
Quadrotor with Robotic Manipulator (QRM) systems are mechanically coupled, differentially flat systems in which an $m$-link manipulator is ideally mounted at the center of mass (\textit{CoM}) of the aerial platform. Such systems possess $m+6$ degrees of freedom (\textit{DoF}), enabling independent regulation of the platform’s translational and rotational motions, together with the manipulator configuration. This increased actuation capability allows QRM systems to perform versatile manipulation and interaction tasks across multiple environments, including ground-based operations, surface interactions on water, and fully aerial missions. This paper contributes the following points to the aerial manipulator for transportation missions:
\begin{enumerate}
    \item
    The first objective of this paper is to establish a mathematical model of the QRM system. To address this, we employ a passive decomposition strategy inspired by \cite{6907674}, in which the system's kinetic energy is decoupled into translational (\textit{CoM} position) and rotational (shape) components. This dynamic decoupling ensures that the system \textit{CoM} acceleration is entirely independent of the RM state variables. While the translational control input remains structurally dependent on the quadrotor's attitude (roll, pitch, yaw) for thrust vectoring, the translation dynamics are freed from the complex, nonlinear inertial coupling forces generated by the RM's motion. This yields two distinct subsystems: a translational `locked' system and an internal `shape' system, where internal motion is rendered gravity-free and structurally non-disruptive to the \textit{CoM} position tracking and its stability.
    \item The second contribution is the system subjected to external disturbances, such as wind or environmental forces acting on the quadrotor's translational position. Integrating these bounded disturbances into the reduced/decomposition model, we present a rigorous analysis of how they propagate through the system's dynamic coupling into the internal configuration variables (roll, pitch, yaw, joint angles). This formulation also accurately captures the indirect effects of disturbances on end-effector precision, which is critical for high-stakes applications, including aerial firefighting, precision spraying, and defense-related maneuvers.
    \item Building upon this disturbance analysis, we propose a continuous nonlinear disturbance observer-based feedback control (NDOBC) law to compensate for these propagated perturbations, separately for the translational and internal configuration dynamics of the perturbed system, ensuring stable motion of translation, quadrotor attitude, and the RM joint angle.
    \item Finally, we present the simulations validation of the proposed controller through two cases: (a) the stabilization/tracking of the quadrotor and arm reference position, and (b) aerial demonstration of a pick-and-place scenario. In case (b), trajectory generation of quintic polynomials facilitates seamless pick-and-place maneuvers and robust end-effector trajectory tracking. 
\end{enumerate}
To present our findings, the first section \ref{System modelling} introduces mathematical modeling for the coupled system (standard dynamics) and the decoupled dynamics. This serves as the basis for Section \ref{Control design}, where the control design and stability proofs are derived. These theoretical results are then validated in Section \ref{sec: simulation} via point stabilization and pick-and-place mission simulations. The paper concludes with Appendix \ref{append} detailing the mathematical preliminaries supporting the control and modeling sections. 
 \section{System Modelling} \label{System modelling}
\subsection{Configuration Space and Coordinates}
Consider a QRM System consisting of a quadrotor of mass $m_0$ and a one-link RM of length $ \ell $ and mass $m_1$ with a gripper. The RM is mounted directly beneath the \textit{CoM} of the quadrotor system. Assuming a rigid gripper, the RM provides two DoF via two revolute joints: one rotates about the $z$-axis and the second about the $y$-axis, as shown in Fig. (\ref{fig:QRMS}). To describe this multi-body system, consider three primary reference frames: $\mathcal{F}_I$, $\mathcal{F}^b_Q$, and $\mathcal{F}^b_M$, which denote the inertial reference frame, the quadrotor body-fixed frame, and the RM body-fixed frame, respectively. \\
The quadrotor pose is defined on the configuration space ${SE}(3) = \mathbb{R}^3 \rtimes SO(3)$, known as the special Euclidean group, where $\mathbb{R}^{3}$ represents the translational positions and $SO(3)$ denotes the special orthogonal group of rotations. The pose is represented by the pair $(\eta_0,R_0) \in SE(3)$. The translational component $\eta_0 = [x\;\;y\;\;z]^\top$ corresponds to the position of the quadrotor as $CoM$ expressed in the inertial frame. To describe the vehicle's rotational component, we adopt the $ZYX$ Euler angle convention from \cite{murray2017mathematical}. Under this convention, the attitude $R_0$ is parameterized by Euler angles as $\Phi = [\phi_r\;\;\phi_p\;\;\phi_y]^\top$, where $\phi_r, \phi_p$ and $\phi_y$ denote the roll, pitch, and yaw angles, respectively. Finally, $\Theta = [\theta_1\;\;\theta_2]^\top$ corresponds to the manipulator joint angles in the configuration space $\mathcal{S}^2 \subseteq \mathcal{S}^1 \times \mathcal{S}^1 $. Hence, the overall configuration space of the multi-body QRM system is given on the manifold $\mathcal{M} = \mathbb{R}^3 \rtimes {SO}(3)\times \mathcal{S}^2$. Consequently, the generalized coordinate vector $q$ describing the full configuration of the QRM system is defined as:
\begin{equation}
  q := [\eta_0^\top \;\; \Phi^\top \;\; \Theta^\top]^\top \in \mathcal{M},  
\end{equation}
with the corresponding velocity vector
\begin{equation}
    \dot q := [\dot \eta_0^\top\;\; \dot \Phi^\top\;\; \dot \Theta^\top]^\top \in T_q\mathcal{M},
\end{equation}
Here, $T_q\mathcal{M}$, is the velocity(tangent) space of the manifold $\mathcal{M}$ at $q$, and $\dot \eta_0 = [\dot x \;\; \dot y \;\; \dot z]^\top\in \mathbb{R}^3$ is the quadrotor platform's $CoM$ velocity in the inertial frame ${\mathcal{F}_I}$; $\dot\Phi = [\dot\phi_r \;\;\dot\phi_p \;\; \dot\phi_y]^\top \in \mathbb{R}^3$ is the Euler rates of the quadrotor, and $\dot {\Theta} = [\dot \theta_1 \;\; \dot \theta_2]^\top \in \mathbb{R}^2$ is the joint angle velocity of the RM. The complete schematic diagram is depicted in Fig. \ref{fig:QRMS}.
\begin{figure}
    \centering
    \includegraphics[width=1.1\linewidth]{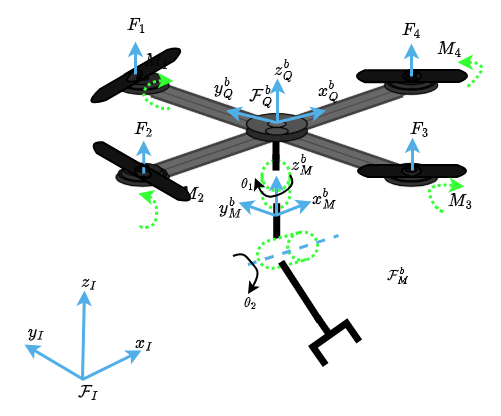}
    \caption{ Schematic of Quadrotor Robotic Manipulator (QRM) system and two Dynamixel servos used for manipulator articulation.}
    \label{fig:QRMS}
\end{figure}
\subsection{Lagrangian dynamics of QRM system}
Let the $CoM$ position of the robotic arm in \{$\mathcal{F}_I$\}, denoted by $\eta_{1}$, be expressed as the sum of the quadrotor's $CoM$ $\eta_{0}$ in \{$\mathcal{F}_I$\} and the RM joint $CoM$ position in \{$\mathcal{F}_I$\}.
\begin{equation} \label{mani_trans}
    \eta_{1} = \eta_0 + R_0 \eta_{1}^b     
\end{equation}
where the vector $\eta_{1}^b \in \mathbb{S}^{2}$  is defined by the link's geometry and the joint angles $\theta_1$ and $\theta_2$, and is formulated as:
\begin{equation}  \label{body_mani}
   {\eta}_{1}^b = R_{1}^b[0\;\;0\;\;-l/2]^\top
\end{equation}
where $R_{1}^b = R_z(\theta_1)R_y(\theta_2)$ is the rotation of the RM in ${\mathcal{F}^b_Q}$.
Define $\eta_{qm} \in \mathbb{R}^3$ as the $CoM$ of the combined QRM system in the $\mathcal{F}_I$. Following the principle of a composite system of two rigid bodies, it is formulated as:
\begin{equation} \label{QRM pose}
    \eta_{qm} =  \eta_{0} + R_0(\Phi) \eta^b_1 \left(\frac{m_1}{m_0 + m_1}\right)
\end{equation}
The linear and angular velocities of the RM are mapped from joint space to task space through the Jacobian matrix, which serves as a linear transformation relating joint rates to the differential translation and rotation of the $CoM$ of the RM.
These velocities are mapped from the joint velocities $\dot \Theta$ through the manipulator Jacobians $J_{v_1}$ and $J_{\omega_1}$, such that
\begin{subequations}\label{equ:link_velocities}
\begin{align}
    \dot{\eta}_{1}^b &= J_{v_1}(\Theta) \dot{\Theta} \label{equ:vel_linear}, \\
    \omega_{1}^b &= J_{\omega_1}(\Theta) \dot{\Theta} \label{equ:vel_angular},
\end{align}
\end{subequations}
where 
\begin{equation*}
\label{eq:trans_jac}
J_{v_1}(\Theta)  =
\begin{bmatrix}
\frac{l}{2}\sin \theta_1 \sin \theta_2 & -\frac{l}{2}\cos \theta_1 \cos \theta_2\\
-\frac{l}{2}\cos \theta_1 \sin \theta_2 & -\frac{l}{2}\sin \theta_1 \cos \theta_2\\
0 &\frac{l}{2}\sin \theta_2
\end{bmatrix} ,
\end{equation*}
\begin{equation*}
 \label{eq:rot_jac}
J_{\omega_1}(\Theta)  =
\begin{bmatrix}
0 & -\sin \theta_1 \\
0 & \cos \theta_1 \\
1 & 0
\end{bmatrix}.
\end{equation*}
Again, by differentiating equation \eqref{mani_trans} and using equation \eqref{equ:vel_linear}, we obtain the $CoM$ linear velocity of the RM,
\begin{equation}\label{eq:lin_ine_vel}
   \begin{aligned}
       \dot\eta_{1} &= \dot\eta_0 + \dot R_0 \eta_{1}^b + R_0 J_{v_1} \dot \Theta.
   \end{aligned}
\end{equation}
Let $\omega_0$ be the angular velocities of the quadrotor in the ${\mathcal{F}_I}$.
Therefore, the standard Kinematic identity for the rigid body is:
\begin{equation} \label{eq:kine_iden}
        \dot{R}_0 = \widehat{{\omega}}_0 R_0,
\end{equation}
$\omega_0$ can be further expressed as an angular rate; this yields:
\begin{equation} \label{equ:angu_velocity}
    \omega_0 = T(\Phi) \dot\Phi\ ,
\end{equation}
where
\[T(\Phi)  = \begin{bmatrix}
    1 & 0 & -\sin(\phi_p) \\
    0 & \cos(\phi_r) & \sin(\phi_r)\cos(\phi_p) \\
    0 & -\sin(\phi_r) & \cos(\phi_r) \cos(\phi_p)
    \end{bmatrix},\]
and $\widehat\omega_0$ is the skew symmetric matrix of the vector $\omega_0$.
From equation \eqref{eq:kine_iden}, we obtain,
\[\dot{R}_0 \, {\eta}^{b}_{1} = \hat{{\omega}}_0 R_0 {\eta}^{b}_{1}.\]
Since $\hat{a}{b} = {a} \times {b}$ and ${a} \times{b} = - {b} \times {a}$,
\begin{equation} \label{kinematic relation}
      \hat{{\omega}}_0R_0 {\eta}^{b}_{1} 
    = {\omega}_0 \times (R_0 {\eta}^{b}_{1})
    = - \widehat{{R_0} {\eta}^{b}_{1}} {\omega}_0.
\end{equation}
Here $\widehat{{R_0} {\eta}^{b}_{1}}$ refers to the position of RM in ${\mathcal{F}_I}$.
Substituting the kinematic relation \eqref{kinematic relation} into \eqref{eq:lin_ine_vel} yields the inertial linear velocity of the RM's $CoM$:
\begin{equation}  \label{manipulator translation}
    \dot{{\eta}}_{1} = \dot{{\eta}}_0 
    - \widehat{R_0{\eta}^{b}_{1}}{\omega}_0 
    + R_0 J_{v_1} \dot{\Theta}.
\end{equation}
Moreover, the angular velocity of the RM in the \{$\mathcal{F_I}$\} can be calculated from equation \eqref{equ:vel_angular} and equation \eqref{equ:angu_velocity}, which gives
\begin{equation} \label{eq:ang_ine_vel}
    \omega_1 = \omega_0 + R_0 J_{\omega_1}\dot\Theta.
\end{equation}
To formulate the system dynamics via the Euler-Lagrange framework, we define the kinetic energy of the QRM system, which is equal to the sum of the quadrotor kinetic energy and the kinetic energy of RM in the \{$\mathcal{F_I}$\}.
The kinetic energy of the quadrotor is calculated from the quadrotor's linear velocity $\dot \eta_0$ and angular rates $\dot \Phi$,
\begin{equation}
\begin{aligned}
KE_{quad} &= \frac{1}{2} \left \langle \dot{\eta_0}, m_0 \dot{\eta_0} \right \rangle + \frac{1}{2}\left \langle \mathbb{I}_0 \omega^b_0, \omega^b_0  \right \rangle, \\ 
      &= \frac{1}{2} \dot{\eta_0}^\top m_0 \dot{\eta_0} + \frac{1}{2} \dot{\Phi}^\top T^\top R_0 \mathbb{I}_0 R_0^\top T \dot{\Phi}.
\end{aligned}
\end{equation}
While the kinetic energy of the RM can be obtained from the linear and angular velocities of the RM, as given by equations \eqref{manipulator translation}-\eqref{eq:ang_ine_vel},
\begin{equation}
    KE_{Mani} = \frac{1}{2} m_1 \dot{\eta_1}^\top {\dot{\eta_1}} 
+\frac{1}{2} \omega_1^\top R_1^{\top} \mathbb{I}_1 R_1 \omega_1.
\end{equation}
Hence, the total kinetic energy of the QRM system is 
\begin{equation} \label{K.E.}
\begin{aligned}
     K.E. =& \frac{1}{2} \dot{\eta_0}^\top m_0 \dot{\eta_0} + \frac{1}{2} \dot{\Phi}^\top T^{\top} R_0 \mathbb{I}_0 R_0^\top T \dot{\Phi}
+ \frac{1}{2} m_1 \dot{\eta_1}^\top {\dot{\eta_1}} \\
&+\frac{1}{2} \omega_1^\top R_1^{\top} \mathbb{I}_1 R_1 \omega_1.
\end{aligned}
\end{equation}
In the above equation, $R_0$ = $R_z({\phi_{y}}) R_y({\phi_{p}}) R_x({\phi_{r}})$ and  
$R_1 = R_z({\phi_{y}}) R_y({\phi_{p}}) R_x({\phi_{r}})R_z({\theta_{1}}) R_y({\theta_{2}})$ represent the quadrotor and RM rotation in the inertial frame, respectively. The matrices $\mathbb{I}_0$ and $\mathbb{I}_1$ are the moment of inertia tensors of the quadrotor and RM system in their respective frames. 
Consequently, the expression for kinetic energy in equation \eqref{K.E.} can be reformulated as: 
\begin{equation} \label{final KE}
    K.E. = \frac{1}{2}\dot{q}^\top M(q)\dot{q},
\end{equation}
where $M$ is called the total mass matrix and is given as:
\begin{equation} \label{eq:mass_matrix}
     M(q) = \begin{bmatrix}
     M_{11} & M_{12} & M_{13} \\
     M_{12}^\top & M_{22} & M_{23} \\
     M_{13}^{\top} & M_{23}^\top & M_{33}
 \end{bmatrix},
\end{equation}
 where, 
 \begin{equation*}
     \begin{aligned}
         M_{11} &=  m_{qm}I_3,\\
      M_{12} &= - m_1 \widehat{R_0\eta_{1}^{b}}T,\\
      M_{13} &= m_1R_0J_{v_1},\\
      M_{22} &=T^{\top}R_0\mathbb{I}_0R_0^{\top}T + m_1 T^\top \widehat{(R_0 \eta_{1}^b)}^\top \widehat{(R_0 \eta_{1}^b)}T+ \\& \quad \quad \hspace{3cm}T^\top R_0 R_{1}\mathbb{I}_1R_{1}^{\top}R_0^{\top}T,  \\
M_{23} &= -m_1T^\top \widehat{(R_0 \eta_{1}^b)}^\top , R_0J_{v_1} + T^{\top}R_0 R_{1}\mathbb{I}_1 R_{1}^\top J_{\omega_1}, \\
M_{31} &= M_{13}^{\top}, \\
M_{32} &= M_{23}^{\top},\\
M_{33} &= m_1 J_{v_1}^{\top}J_{v_1}+J_{\omega_1}^{\top}R_{1}\mathbb{I}_1R_{1}^{\top}J_{\omega_1},
     \end{aligned}
 \end{equation*}
here, $I_3 \in \mathbb{R}^{3\times 3}$ denotes the identity matrix and $m_{qm}$ represents the total mass of the system, defined as $m_0 + m_1$.
The QRM system is under the effect of gravity along the \{$\mathcal{F_I}$\}’s negative $Z$-direction.  
Thus, the gravitational potential energy is,
\begin{equation} \label{potential energy}
U(q) := - \sum_{i=0}^1 m_i g \langle \eta_i, \mathbf{e}_3 \rangle,
\end{equation}
where g is the acceleration due to gravity, which is taken as 9.8 m/s$^2$, and $\mathbf{e}_3$ is a $3\times1$ unit vector in the $z$-direction.\\
It suffices to consider the time derivative of a scalar function or vector function, which is denoted by $\dot{()}$ for the first derivative and $\ddot{()}$ for the second derivative. Higher-order derivatives are represented by $()^{(i)}$, denoting the $i^{th}$ derivative with respect to time.
Using the kinetic energy given by equation \eqref{final KE} and the gravitational potential energy given by equation \eqref{potential energy}, the Euler-Lagrange dynamics of the QRM system is
\begin{equation} \label{dynamics}
 M(q)\ddot{q} + C(q,\dot q)\dot{q} + G(q) = \tau + \tau_d,   
\end{equation}
where:
\begin{itemize}
    \item $C(q,\dot{q}) \in \mathbb{R}^{8\times 8}$ is the Coriolis matrix whose elements are defined via the Christoffel symbols of the first kind as $C_{kj} = \sum_{i=1}^8 \frac{1}{2} \left( \frac{\partial m_{kj}}{\partial q_i} + \frac{\partial m_{ki}}{\partial q_j} - \frac{\partial m_{ij}}{\partial q_k} \right) $, ensuring that the matrix $\dot{M}(q) - 2C(q,\dot{q})$ is skew-symmetric \cite{murray2017mathematical}.
    \item $G(q) = \partial U(q) / \partial q$ is the gravitational force, 
    given by
\begin{equation} \label{G(q)}
    \begin{aligned}
G(q) = g\begin{bmatrix}
 m_{qm}\mathbf{e}_3^\top &
m_1 \frac{\partial\langle\eta_1,\mathbf{e}_3\rangle}{\partial \zeta}
\end{bmatrix}^\top,
\end{aligned}
\end{equation}
where $q[i]$ is meant for the $i^{th}$ element of vector $q$.
\item The control input vector $\tau \in \mathbb{R}^8$ is mapped from a set of independent actuators defined by $u =
        \left(u_{t}R_0\mathbf{e}_3, u_{\Phi}, u_{\Theta}
   \right)$ $\in \mathbb{R}^8$. In this formulation,  $u_{t} \in \mathbb{R}$ represents the total thrust of the quadrotor, and $u_{t}R_0\mathbf{e}_3 \in \mathbb{R}^3$ represents the thrust vector in the $\mathcal{F}_I$. $u_{\Phi} \in \mathbb{R}^3$ denotes quadrotor attitude control, and $u_{\Theta}\in \mathbb{R}^2$ corresponds to the joint torques of the robotic manipulator (RM). 
    \item $\tau_d = \begin{bmatrix}
        [{\tau_{dL}}]_{3 \times 1} \\ [\tau_{dS}]_{5 \times 1}
    \end{bmatrix}\in \mathbb{R}^8$ is the external disturbance.
\end{itemize}
\begin{assumption}\label{assump_1}
    External disturbances is constant, that is, $\dot{\tau}_{dL} = 0$ with $\| \tau_{dL}\| \leq \sigma$.
\end{assumption}
\subsection{Passivity-Based Decomposition of the Configuration Manifold}
A mechanical system defined by generalized coordinates $q$ in the configuration manifold $\mathcal{M}$ and velocities $\dot q$ in the velocity space of $ \mathcal{M}$ is classified as passive if its total stored energy grows no faster than the external power supply. For the multi-body dynamics in equation \eqref{dynamics}, this energy is characterized by the inertia matrix $M(q)$, which induces a Riemannian metric on the manifold $\mathcal{M}$. By using the Riemannian metric, we can form  $M(q)$-orthogonal subspaces within the velocity space to effectively decouple the system’s kinetic energy mathematically. This decomposition ensures that the RM motion does not act as an internal disturbance to the quadrotor's translational path, thereby isolating the energy dynamics of each subsystem.
To facilitate the analysis of the decoupled system's dynamics, the configuration vector $q$ is partitioned into locked and shape variables as follows: \begin{equation} q = [\eta_0, \zeta]^\top \in \mathcal{M}, \end{equation} where $\eta_0 \in \mathbb{R}^3$ represents the previously defined $CoM$ position of the quadrotor. The vector of internal configuration variables or shape variables is introduced as \begin{equation} \zeta = [\phi_r\;\; \phi_p\;\; \phi_y\;\; \theta_1\;\; \theta_2]^\top, \end{equation} which collects the quadrotor's attitude and the joint angles of the RM. Consequently, the mass matrix defined in equation \eqref{eq:mass_matrix} can be expressed in the following block-partitioned form: 
\begin{equation} \label{eq: block mass matrix}
    M(q) = \begin{bmatrix}
        \tilde{M}_{11} &  \tilde{M}_{12}\\
        {\tilde{M}_{12}}^\top & \tilde{M}_{22}
    \end{bmatrix},
\end{equation}
where, 
\begin{align*}
    \tilde{M}_{11} &= m_{qm}I_3,\\
    \tilde{M}_{12} &= \begin{bmatrix}
        M_{12} &  M_{13}
    \end{bmatrix},\\
    \tilde{M}_{22} &= \begin{bmatrix}
        M_{22} &  M_{23}\\
        M_{23}^\top &  M_{33}
    \end{bmatrix}.
\end{align*}
We define a smooth mapping
\begin{equation}
f : \mathcal M \to SO(3)\times \mathcal{S}^2, \qquad f(q) = \zeta,
\end{equation}
 which isolates the attitude and joint variables from the translational coordinates of the QRM system. This can be expressed using a selection matrix $H \in \mathbb{R}^{5 \times 8}$, such that,
\[\zeta = Hq, \] 
where H is defined as the selection matrix $H = \begin{bmatrix}
    0_{5\times 3} & I_{5 \times 5}
\end{bmatrix}$. Given that $f$ is a linear mapping of q, its Jacobian is a constant matrix $Df = \frac{\partial f}{\partial q} = H$. By the Rank-Nullity Theorem, given that $dim(ker(H))=3$ and the configuration space $\mathcal{M}$ has dimension 8, this yields that the dimension of the image $im(H)$ is exactly 5. Because the image spans the entire co-domain $\mathbb R^{5}$, the linear mapping $f$ is a smooth submersion.

Proceed by defining the level set of $f$ such that it is the set of all configurations $q$ that satisfy the state constraint:
\begin{equation}
\mathcal M_f = \{ q \in \mathcal M \mid f(q) = c \}, \, \text{for some $c \in \mathbb{R}^5$}.
\end{equation} 
Geometrically, this level set represents a 3-dimensional smooth submanifold within an 8-dimensional configuration manifold $\mathcal{M}$.
Since $f$ is a submersion, $M_f$ is an embedded submanifold of $\mathcal{M}$. It's velocity subspace at $q$, denoted by $T_q \mathcal M_f$, is given by the kernel of  differential of $f$ at $q$:
\begin{equation}
T_q \mathcal M_f = \ker Df(q).
\end{equation}
We call it the tangential (locked) distribution, and it is expressed as,
\begin{equation}
\mathcal{D}_q := T_q \mathcal M_f
= \{ \dot q \in T_q\mathcal M \mid D_qf(q)\dot q = 0 \}.
\end{equation}
Vectors in $\mathcal{D}_q$ correspond to instantaneous motions that preserve the internal shape variables, i.e., $\dot \zeta = 0$. These directions correspond to pure translational motion of the QRM system, with the internal configuration locked. Hence, from the \eqref{QRM pose}, it is evident that $\dot \eta_0 = \dot \eta_{qm}$, and this renders the manifold $\mathcal{D}_q$ to be re-written as,
\begin{equation} \label{tangential distribution}
    \mathcal{D}_q = \left\{ \begin{bmatrix} I_3 & \mathbf{0}_{3 \times 5} \end{bmatrix} \begin{bmatrix} \dot{\eta}_{qm} \\ \dot{\zeta} \end{bmatrix} \in \mathbb{R}^3 \;\middle|\; \dot{\eta}_{qm} \in \mathbb{R}^3, \dot{\zeta} = \mathbf{0}_{5\times1} \right\}
\end{equation},
where $\dot{\eta}_{qm} = [\dot{x}_{qm}, \dot{y}_{qm}, \dot{z}_{qm}]^\top$ denotes the velocity of the $CoM$ of the QRM system in  $\mathcal{F}_{I}$.
Let ${u, v}$ $\in$ $T_q\mathcal M$. Then the kinetic energy induces a Riemannian metric on $T_q\mathcal M$:
\begin{equation} \label{M-orthogonality}
\langle \! \langle u, v \rangle \!\rangle_q := u^\top M(q) v.
\end{equation}
Using this metric, we define the \emph{normal (shape) distribution} as the orthogonal complement of $\mathcal{D}_q$:
\begin{equation} \label{normal distribution}
\mathcal{D}_q^\perp
:= \{ v \in T_q\mathcal M \mid \langle \! \langle v, u \rangle \!\rangle_q = {0}, \ \forall u \in \mathcal{D}_q \}.
\end{equation}
This definition ensures that motions in $\mathcal{D}^\perp_q$ are orthogonal to locked translational motions $\mathcal{D}_q$.
To calculate set $\mathcal{D}^\perp_q$, let any $v = \begin{bmatrix} \dot \eta_0 \\ \dot \zeta \end{bmatrix} \in T_q\mathcal{M}$ and $u \in \mathcal{D}_q$;
thus, the orthogonality condition given by equation \eqref{M-orthogonality} yields,
\begin{equation*}
    0 = v^\top M(q)\,u, \;\; \forall u=\begin{bmatrix}\dot \eta_{qm}\\0\end{bmatrix}\in\mathcal{D}.
\end{equation*}
This gives,
\begin{equation} \label{eq:momentum}
    m_{qm} \dot \eta_0 + \tilde M_{12} \dot\zeta = 0
\quad\Longrightarrow\quad
\dot\eta_0 = -\frac{1}{m_{qm}} \tilde M_{12}\, \dot\zeta.
\end{equation}
Defining
\begin{equation} \label{N and M12}
 N(\zeta) := -\tfrac{1}{m_{qm}} \tilde M_{12}(\zeta).
 \end{equation}
and substituting the vector $v$, we get
\begin{figure}
    \centering
\includegraphics[width=1\linewidth]{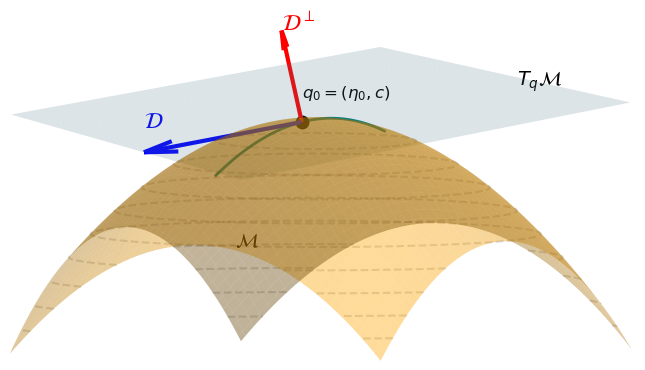}
    \caption{Geometric interpretation of locked and shape systems in an aerial manipulator}
    \label{Geometry_locked_shape}
\end{figure}
\begin{equation}
    v = \begin{bmatrix} N(\zeta) \\ {I}_5 \end{bmatrix} \dot \zeta \in \mathcal{D}^\perp.
\end{equation}
Finally, It is evident that, 
\begin{align}
\mathcal{D}_q \cap \mathcal{D}_q^\perp &= \{0\}, \\
\dim \mathcal{D}_q + \dim \mathcal{D}_q^\perp&= \dim T_q\mathcal M.
\end{align}
This implies that the tangent space admits the orthogonal direct sum decomposition
\begin{equation} \label{decomposition}
T_q\mathcal M = \mathcal{D}_q \oplus \mathcal{D}_q^\perp.
\end{equation}
This decomposition is referred to as energy based decomposition, as it is directly induced by the kinetic energy of the system. The geometric relationship of this decomposition between the shape space and the locked manifold as a tangent and normal direction in the associated tangent space is illustrated in Fig. \ref{Geometry_locked_shape}.
Therefore, the total velocity vector is partitioned as:
\begin{equation} \label{Decomposition}
        \dot q = \begin{bmatrix}
    I_3 & N(\zeta) \\
    0 &   I_5
\end{bmatrix} \begin{bmatrix}
    \dot \eta_{qm} \\
    \dot \zeta
\end{bmatrix} =  \mathcal{B}(\zeta) \nu 
\end{equation}.
Along the tangent space $\mathcal{D}_q$, the translational velocity of the quadrotor coincides with the velocity of the $CoM$ of the QRM system, which results in the multi-body system acting as a single rigid body. Consequently, the total gravitational potential energy of the system becomes uniquely parameterized by the $CoM$ position. Under this projection, gravitational forces map exclusively into the translation acceleration dynamics of the QRM $CoM$. These internal motions are orthogonal to the locked translational motion; therefore, gravity does not appear explicitly in their dynamics. This decomposition yields a structure-preserving separation between $CoM$ of the QRM system and internal shape dynamics.
\subsection{Decoupled dynamics}
From the decomposition  \eqref{Decomposition}, the Euler-Lagrange equation \eqref{dynamics} can be derived as follows: 
\begin{equation*}
    M(q)\mathcal{B}(\zeta)\dot{\nu} + \Bigl(M(q){\dot {\mathcal{B}}(\zeta)} + C(\zeta,\dot{\zeta})\mathcal{B}(\zeta)\Bigr)\nu + G = \tau
\end{equation*}
Multiply by ${\mathcal{B}(\zeta)}^{\top}$ ,
\begin{equation*}
\begin{aligned}
    {\mathcal{B}(\zeta)}^{\top}M(q)\mathcal{B}(\zeta)\dot{\nu} + {\mathcal{B}(\zeta)}^{\top}\Bigl(M(q)\dot{\mathcal{B}}(\zeta) + C(\zeta,\dot{\zeta})\mathcal{B}(\zeta)\Bigr)\nu +\\{\mathcal{B}(\zeta)}^{\top}G = {\mathcal{B}(\zeta)}^{\top}\tau.
\end{aligned}
\end{equation*}
By simple matrix multiplication and using the \eqref{N and M12}, the above equation can be described as ,
\begin{equation} \label{new dynamics}
    \bar{M}(\zeta) \dot{\nu} + \bar{C}\nu + \bar{G} = \bar{\tau} + \bar{\tau}_d,
\end{equation} 
    where 
    \begin{equation} 
        \bar{M}(\zeta) = \begin{bmatrix}
    m_{qm} I_3 & \mathbf{0}_{3\times5}\\
    \mathbf{0}^\top_{5\times3} & N^\top \tilde M_{12}+\tilde M_{22} 
\end{bmatrix}.
\end{equation}
    Define $N^\top \tilde M_{12}+ \tilde M_{22} = M_S$; then $ \bar M(\zeta)$ can be re-written as:
    \begin{equation} \label{M_bar}
        \bar{M}(\zeta) = \begin{bmatrix}
    m_{qm} I_3 & \mathbf{0}\\
    \mathbf{0}^\top & M_S 
\end{bmatrix},
\end{equation}
\begin{equation} \label{G_bar}
   \bar G = \begin{bmatrix}
        m_{qm}g\mathbf{e}_3 \\ \mathbf{0}_{5\times1}
    \end{bmatrix},
\end{equation} and
\begin{equation}
    \bar \tau = \begin{bmatrix}
        u_tR_0\mathbf{e}_3 \\ N(\zeta)^\top u_tR_0\mathbf{e}_3 + \begin{bmatrix}
            u_\Phi \\ u_\Theta
        \end{bmatrix}
    \end{bmatrix}, \quad \bar \tau_d = \begin{bmatrix}
        \tau_{dL} \\ N(\zeta)^\top \tau_{dL} + \tau_{dS}
    \end{bmatrix}.
\end{equation}
We consider an external disturbance ($\tau_{dL}$) based on Assumption -\ref{assump_1} acting on the quadrotor’s translational position, while assuming no direct external disturbances act on the internal arm variables $\zeta$ (i.e., $\tau_{dS} =0$). However, due to the coupled dynamics of the Aerial Manipulator (UAV-Arm system), this constant task-space disturbance propagates into the internal configuration as a time-varying disturbance term, $N^\top\tau_{dL}$. This accurately reflects real-world scenarios where external wind or environmental forces acting on the quadrotor body indirectly perturb the arm's joint stability through the system's kinematic coupling. Then
\begin{equation}
    \bar \tau_d = \begin{bmatrix}
        \tau_{dL} \\ N^\top\tau_{dL}
    \end{bmatrix}.
\end{equation}
In this derivation $N^\top\tau_{dL}$, maps task-space disturbances into the shape-space manifold.
By focusing on the estimation of $\tau_{dL}$, we avoid the redundancy of defining an independent shape-space disturbance $\tau_{dS}$. This approach acknowledges that internal perturbations in quadrotor attitude and joint angles are primarily induced by the transmission of external task-space loads through the system's time-varying geometry $(N(\zeta))$, rather than originating from isolated internal sources.
For the rigorous calculation of $\bar M (\zeta)$ and $\bar G$, refer to the appendices \ref{bar M} and \ref{bar G}.
We now calculate $\bar C$ using the skew-symmetry property of a simple mechanical system, rather than directly computing ${\mathcal{B}(\zeta)}^{\top}\Bigl(M(q)\dot{\mathcal{B}}(\zeta) + C(\zeta,\dot{\zeta})\mathcal{B}(\zeta)\Bigr)$. A fundamental property of Lagrangian dynamics, given by \eqref{dynamics}, states that $\dot{M}(q)-2C$ is skew-symmetric. Using this property, we show that ${\dot{\bar{M}}(\zeta)}-2\bar{C}$ is also skew-symmetric.
\begin{lemma} \label{barM - 2C skew symmetric}
    If $\dot M - 2C$ is skew symmetric according to the standard dynamics of \eqref{dynamics}, then $\dot {\bar M} - 2\bar C$ will also be skew symmetric from the new decoupled dynamics \eqref{new dynamics}.
\end{lemma}
\begin{proof} The fundamental property of Lagrangian dynamics in standard dynamics is given by \eqref{dynamics}, which states that $\dot{M}(q)-2C$ is skew-symmetric. So, without loss of generality, let us consider $\mathcal{B}(\zeta) = \mathcal{B}$ and $\dot {\bar M}(\zeta) = \dot{\bar M}$. Therefore
 \begin{equation}
\dot{\bar{M}}-2\bar{C}
=
\frac{d}{dt}(\mathcal{B}^\top M\mathcal{B})-2\mathcal{B}^\top (M\dot{\mathcal{B}}+C\mathcal{B}).
\end{equation}
Applying the product rule to the first term:
\begin{equation}
\begin{aligned}
\dot{\bar{M}}-2\bar{C}
&=
{\mathcal{\dot B}}^\top M\mathcal{B} + \mathcal{B}^\top \dot{M}\mathcal{B} + \mathcal{B}^\top M\dot{\mathcal{B}}\\
&\quad
-2\mathcal{B}^\top M{\mathcal{\dot B}}-2\mathcal{B}^\top C\mathcal{B}.
\end{aligned}
\end{equation}
Combining terms with respect to $\dot M - 2C$:
\begin{equation}
\dot{\bar{M}}-2\bar{C}
=
\mathcal{B}^\top(\dot{M}-2C)\mathcal{B}
+
(\dot{\mathcal{B}}^\top M\mathcal{B}-\mathcal{B}^\top M\dot{\mathcal{B}}).
\end{equation}
Let
$
A=\mathcal{B}^\top (\dot{M}-2C)\mathcal{B}, \,\,  
B=\dot{\mathcal{B}}^\top M\mathcal{B}-\mathcal{B}^\top M\dot{\mathcal{B}}$.
Since $(\dot{M}-2C)$ is skew-symmetric, then
$
A^\top
\mathcal{B}^\top(-(\dot{M}-2C))\mathcal{B}
=
-A$. Thus, $A$ is skew-symmetric; similarly, 
\[
B^\top
=
(\mathcal{B}^\top M\dot{\mathcal{B}})^\top -(\dot{\mathcal{B}}^\top M\mathcal{B})^\top = - B.\]
This yields that $B$ is also skew-symmetric.
Since the sum of two skew-symmetric matrices is skew-symmetric, we conclude that
\[\dot{\bar{M}}-2\bar{C},\] is skew-symmetric.
\end{proof}
 By Lemma \eqref{barM - 2C skew symmetric}, it suffices to compute $\bar{C}$ from $\bar{M}$ using the Christoffel symbols.
From the decoupled dynamics of \eqref{new dynamics}, we have $\bar M(\zeta)$ from equation \eqref{M_bar}.
Let $\bar{C}$ be partitioned as
\[
\bar{C}
= \begin{bmatrix}
\bar{C}_{11} & \bar{C}_{12}\\
\bar{C}_{21} & \bar{C}_{22}
\end{bmatrix},
\]
Hence, the Christoffel symbols $C_{ijk}$ associated with the \eqref{M_bar}, implies,
\[\bar{C}_{11}=0,\qquad
\bar{C}_{12}=0,
\qquad
\bar{C}_{21}=0.\]
and for ${i,j,k} \in \zeta$, the derivatives $\frac{\partial \bar M_{ij}}{\partial \zeta_k}$ are non-zero because the internal inertia $M_S$ depends only on the vector variable $\zeta$. This yields the standard Coriolis matrix after the decomposition is performed. 
\begin{align} \label{bar_C}
\bar{C}=
\begin{bmatrix}
0 & 0\\
0 & C_S(\zeta,\dot{\zeta})
\end{bmatrix},
\end{align}
where $C_S$ satisfies the standard robotic property
\[
\dot{M}_S-2C_S(\zeta, \dot \zeta),
\] is skew-symmetric; which is rigorously proved in the appendix \ref{skew symmetric of dotM_s-2C_s}.\\
 In this passive decomposition framework, the transformation \eqref{decomposition} is designed to decouple the total system.  This physical requirement ensures that there is no Coriolis coupling between the total system translation and the internal configuration variables.
This means that the translation dynamics of the QRM system are given by the locked system, which has no Coriolis terms therein, and the rotational dynamics, referred to as the shape system, is free from the gravity term.  \\
Substituting equations \eqref{M_bar}, \eqref{G_bar}, and \eqref{bar_C} into \eqref{new dynamics}, we can transform the decoupled dynamics \eqref{new dynamics} into two second order systems
    \begin{align} 
    m_{qm} {\ddot \eta}_{qm} + G_L &= \tau_L + \tau_{dL}, \label{eq:locked}\\
    M_S(\zeta){\ddot \zeta} + C_S(\zeta,\dot{\zeta})\dot{\zeta} &= \tau_S + N^\top \tau_{dL}.\label{eq:shape}
\end{align}
where
\begin{enumerate}
    \item The locked system \eqref{eq:locked} describes the $CoM$ 
    dynamics of the QRM system with $\dot{\eta}_{qm} = \dot{\eta}_{0}$, 
    $m_{qm} = m_0 + m_1$, 
    $G_L = -[0\;\;0\;\; m_{qm} g]^T \in \mathbb{R}^3$, 
    and $\tau_L = u_{t} R_0(\Phi)\mathbf{e}_3$, 
    where $\eta_{qm}$ is the total $CoM$ position of the QRM system and $u_{t}$ is the thrust scalar, and $R_0(\Phi)$ is the rotation of the quadrotor.
    \item The shape system \eqref{eq:shape} describes the internal rotational 
    dynamics of $\zeta = [\Phi \;\; \Theta]^\top$ of the QRM system with full actuation 
    $\tau_S = \begin{bmatrix}
         \tau_{s1} \\ \tau_{s2}
    \end{bmatrix}\in \mathbb{R}^{5 \times 1}$, where $\tau_{s1} = -(1/m_{qm})M^\top_{12}\tau_L+u_\Phi$ and $\tau_{s2} = -(1/m_{qm})M^\top_{13}\tau_L+u_\Theta$.
\end{enumerate}
Here $M_S \in \mathbb{R}^{5\times5}$ 
    is a positive symmetric inertia matrix.\\
These resulting decoupled subsystems, formally defined by \eqref{eq:locked}-\eqref{eq:shape}, significantly reduce controller execution latency, thereby ensuring the numerical robustness and high-frequency stability required for hardware-in-the-loop experiments and deployment on a physical robot.
\begin{remark}
The dynamics described by \eqref{eq:locked}-\eqref{eq:shape} exhibit significant parallels with non-holonomic systems. In the study of non-holonomic mechanics, specifically those characterized by non-integrable constraints, the system's kinetic energy can be analyzed by decomposing the velocity space into a direct product of horizontal and vertical subspaces \cite{Bloch1996}. This structure is often best identified through the framework of Lie groups and their associated Lie algebras \cite{MarsdenRatiuAbraham2022}. Building on this geometric foundation, the tangent bundle $T_q\mathcal{M}$  is partitioned into a tangential (horizontal) distribution $\mathcal{D}_q$, which represents the set of allowable velocities, and a normal (vertical) distribution $\mathcal{D}^\perp_q$, which accounts for the directions restricted by the non-holonomic constraints. As formally established in \eqref{decomposition}, this geometric decomposition effectively captures the full constrained dynamics of the system within a unified manifold representation.
\end{remark} 
\begin{remark}
The quadrotor rotor speeds $\omega_i$ generate the total thrust $u_{t}$ and the body torques $\boldsymbol{\tau} = [\tau_{\phi_r}, \tau_{\phi_p}, \tau_{\phi_y}]^T$. This relationship is mapped via the control allocation matrix such as:
\begin{equation}
\begin{bmatrix}
u_{t}\\
\tau_{\phi_r}\\
\tau_{\phi_p}\\
\tau_{\phi_y}
\end{bmatrix}
=
\begin{bmatrix}
a & a & a & a\\
al & 0 & -al & 0\\
0 & -al & 0 & al\\
b & -b & b & -b
\end{bmatrix}
\begin{bmatrix}
\omega_1^2\\
\omega_2^2\\
\omega_3^2\\
\omega_4^2
\end{bmatrix},
\end{equation}
where $a$ is the thrust coefficient, $b$ is the drag coefficient, and $l$ denotes the length of the arm from the $CoM$ to each rotor.
\end{remark}
\begin{assumption} \label{gain_matrix}
Let $\mathcal{D}_n^+$ is a set of all positive definite diagonal matrices, i.e.,
\[
\mathcal{D}_n^+
=
\left\{
\operatorname{diag}(d_1,\ldots,d_n)\in\mathbb{R}^{n\times n}
\,\middle|\,
d_i>0,\; i=1,\ldots,n
\right\}.
\]
\end{assumption}
\section{Control Design and Stability Analysis}
\label{Control design}
To introduce the state-space representation of dynamics \eqref{eq:locked}–\eqref{eq:shape}, we define the state vectors as follows: 
$x_1 = \eta_{qm}$, $x_2 = \dot\eta_{qm}$,$y_1 = \zeta$, and $y_2 = \dot\zeta$, where $x_1$ and $y_1$ denote the locked and shape positions, respectively, while $x_2$ and $y_2$ denote their corresponding velocities. Consequently, the resulting state-space dynamics are expressed as: 
\begin{align}
\left.
\begin{aligned} \label{eq:locked_ss}
\dot x_1 &= x_2, \\
m_{qm}\dot x_2 &= 
\big(\tau_L - G_L\big) + \tau_{dL}
\end{aligned}
\right\},
\\[6pt]
\left.
\begin{aligned}\label{eq:shape_ss}
\dot y_1 &= y_2, \\
{M_{S}(y_{1})}\dot y_2 &= 
\tau_{S} - C_{S} y_2 + N(y_{1})^\top \tau_{dL}
\end{aligned}
\right\}.
\end{align}
By utilizing the structure of state space dynamics \eqref{eq:locked_ss}-\eqref{eq:shape_ss}, the disturbance observer based control is designed to demonstrate the stability of the system.
\subsection{Problem Statement}
Consider $x_{1d}$ and $y_{1d}$ are the desired $CoM$ $\eta^d_{qm}$ and desired internal configuration $\zeta^d$ of the QRM system, respectively. Moreover, let $\dot x_{1d}$ and $\dot y_{1d}$ be the desired velocities, the same as $\dot \eta^d_{qm}$ and $\dot\zeta^d$, respectively. 
Thus, the tracking errors for the locked and shaped variables are defined as:
\begin{equation} \label{error states}
    \begin{aligned}
        e_{x_1} &= x_1 - x_{1d}, \quad &e_{x_2} &= x_2 - \dot{x}_{1d} \\
        e_{y_1} &= y_1 - y_{1d}, \quad &e_{y_2} &= y_2 - \dot{y}_{1d}.
    \end{aligned}
\end{equation}
We define two filter variables $S_l$ and $S_s$ for the locked and shape systems, respectively, such that
\begin{align} 
    S_l &= e_{x_2} + \lambda_l e_{x_1} \label{locked surface}, \\
    S_s &= e_{y_2} + \lambda_s e_{y_1} \label{shape surface}.
\end{align}
Taking the time derivative 
\begin{align}
    \dot S_l &= \dot e_{x_2} + \lambda_l \dot e_{x_1}, \label{derivative locked surface}\\
    \dot S_s &= \dot e_{y_2} + \lambda_s \dot e_{y_1} \label{derivative shape surface},
\end{align}
where $\lambda_l \in \mathcal{D}_3^+$, $\lambda_s \in \mathcal{D}_5^+$ are based on Assumption-\ref{gain_matrix}.
The system is subject to an unknown constant disturbance vector $\tau_{dL} \in \mathbb{R}^3$ acting on the Locked system (e.g., wind or unmodeled forces). Due to the mechanical coupling defined by the passive decomposition, this disturbance projects into the Shape system via the state-dependent mapping matrix $N(\zeta)\in \mathbb{R}^{5\times 3}$.\\
Standard adaptive control will create an algebraic loop due to the thrust of the quadrotor required to calculate the desired attitude of the quadrotor. To overcome this problem, we employ a Disturbance Observer (DO).\\
Consider an auxiliary state variable $z\in \mathbb{R}^3$; then define the estimate of the disturbance $\hat \tau_{dL}$ such that
\begin{equation}\label{error estimation}
    \hat \tau_{dL} = z+\gamma_lm_{qm}S_l,
\end{equation}
where $\gamma_l>0$ is the observer gain. The auxiliary state $z$ is updated via the following Observer Dynamics 
\begin{equation} \label{OB dynamics}
    \dot z = -\gamma_l(\tau_L - G_L - m_{qm}\ddot x_{1d} + \hat\tau_{dL} + m_{qm}\lambda_l \dot  e_{x_1}).
\end{equation}
Then, the control objective is twofold:
\begin{enumerate}
    \item \textbf{Locked Motion Tracking:} 
    Design a control input $\tau_L$ such that the translational configuration $x_1 \in \mathbb{R}^3$ and its velocity $x_2\in\mathbb{R}^3$ exponentially track the desired smooth trajectory $(x_{1d}, \dot{x}_{1d})$. Specifically, we require:
    \begin{equation}
        \lim_{t \to \infty} \| e_{x_1}(t) \| \rightarrow 0, \quad \lim_{t \to \infty} \| e_{x_2}(t) \| \rightarrow 0.
    \end{equation}
    \item \textbf{Shape system tracking:} 
    Design a control input $\tau_{S}$ such that the shape variables $y_1 \in \mathbb{R}^5 $ converge to a desired configuration $y_{1d}\in \mathbb{R}^5$ and their corresponding velocities $y_{2d}\in \mathbb{R}^5$. That is:
    \begin{equation}
        \lim_{t \to \infty} \| e_{{y}_1}(t) \| \rightarrow 0, \quad \lim_{t \to \infty} \| e_{{y}_2}(t) \| \rightarrow 0.
    \end{equation}
    This ensures the manipulator not only reaches the target posture but also stabilizes without residual oscillation.
    \item \textbf{External disturbance} : With the control inputs $\tau_L$ and $\tau_S$, the estimated external disturbance in the locked system $\hat{\tau}_{dL}\in\mathbb{R}^3$ converges to the actual external disturbance $\tau_{dL}\in\mathbb{R}^3$. Consequently,
$N^\top \hat{\tau}_{dL}\in\mathbb{R}^5$ converges to $N^\top \tau_{dL}\in\mathbb{R}^5$.
\end{enumerate}
Accordingly, the following control laws are proposed:
\begin{align*} 
    u_{ct} &=  \left( G_L - \hat \tau_{dL}+ m_{qm} \ddot{x}_{1d} - m_{qm}\lambda_l\dot e_{x_1} - K_lS_l \right). 
\end{align*}
Therefore, locked control $\tau_L$ will be (see Corollary \ref{tauL = uct}), 
\begin{equation}\label{eq:ps_locked_ctrl}
    \tau_L = u_{ct},
\end{equation}
and the shape control vector $\tau_S$ is designed as,
\begin{equation}\label{eq:ps_shape_ctrl}
\begin{aligned}
\tau_{S}
&=
C_S(y_1,y_2) (-\lambda_se_{y_1} + \dot y_{1_d})- N(y_{1})^\top\hat\tau_{dL} + M_S(y_{1})\ddot y_{1d} - \\
&\quad
M_S(y_1)\lambda_S \dot e_{y_1}- K_s S_s,
\end{aligned}
\end{equation}
where $K_l \in \mathcal{D}_3^+$ and $K_s\in \mathcal{D}_5^+$ are gain matrices.
At last the estimation error is defined as 
\begin{equation} \label{error disturbance}
    \tilde \tau_{dL} = \tau_{dL} - \hat\tau_{dL},
\end{equation}
under the Assumption-\ref{assump_1} $\dot\tau_{dL} = 0$, and $\dot {\hat\tau}_{dL}$ calculated from \eqref{error estimation} gives
\begin{equation}\label{eqn:dist_extimator}
    \dot {\tilde \tau}_{dL} = -\dot z-\gamma_lm_{qm}\dot S_l.
\end{equation}
Later, substituting \eqref{eq:locked_ss},  \eqref{derivative locked surface}, and observer dynamics \eqref{OB dynamics} into \eqref{eqn:dist_extimator} yields this,
\begin{equation*}
    \begin{aligned}
        \dot {\tilde\tau}_{dL} =& \Big(\gamma_l(\tau_L - G_L - m_{qm}\ddot x_{1d} + \hat\tau_{dL} + m_{qm}\lambda_l \dot  e_{x_1}) - \\&\gamma_l (\tau_{L} - G_L + \tau_{dL} - m_{qm}\ddot x_{1d}+m_{qm}\lambda_l\dot e_{x_1})\Big),\\
        =& \gamma_l \hat \tau_{dL} - \gamma_l\tau_{dL}.
    \end{aligned}
\end{equation*}
\begin{equation}
    \dot {\tilde{\tau}}_{dL} = -\gamma_l\tilde \tau_{dL}.
\end{equation}
\begin{assumption}
To ensure the existence of the desired attitude variable $\Phi$ and valid control inputs, we require the desired trajectory to satisfy $x_{1d} \in \mathcal{C}^{4}$ and $y_{1d} \in \mathcal{C}^{2}$. Specifically, $x_{1d}$ is required to be a four-times continuously differentiable function, while $y_{1d}$ is required be at least twice continuously differentiable.
\end{assumption}
\subsection{Desired Attitude Generation}
The quadrotor system is underactuated, as the translational motion along the (X)- and (Y)-axes is coupled with the vehicle attitude. To track the desired CoM position $(\eta_{qm}^d)$ of the QRM system, a virtual control force $(u_{ct}\in\mathbb{R}^3)$ is defined based on the translational tracking error. Assuming a standard (ZYX) Euler-angle convention, the desired roll $(\phi_r^d)$ and pitch $(\phi_p^d)$ angles are obtained by inverting the rotation matrix $R_0$ with respect to the required thrust direction. Consequently, the desired roll and pitch angles can be expressed in terms of the required force
\begin{equation}
\begin{aligned}
 \phi_p^d &= \tan^{-1}\left[\frac{\cos(\phi^d_y)\langle u_{ct},\mathbf{e}_1\rangle + \sin(\phi^d_y)\langle u_{ct},\mathbf{e}_2\rangle}{\langle u_{ct},\mathbf{e}_3\rangle }\right] \label{eq:pitch desired}
\end{aligned},
\end{equation}
\begin{equation}
    \begin{aligned}
        \phi_r^d &= \tan^{-1}\left[\frac{\cos(\phi^d_p)\Bigl(sin(\phi^d_y)\langle u_{ct},\mathbf{e}_1\rangle - \cos(\phi^d_y)\langle u_{ct},\mathbf{e}_2\rangle\Bigr)}{\langle u_{ct},\mathbf{e}_3\rangle }\right] \label{eq:roll desired}
    \end{aligned}.
\end{equation}
While the desired yaw angle $(\phi_y^d)$ can be arbitrarily assigned by the user for task-specific requirements (such as orienting the manipulator).
In the next subsection, the stability of the QRM system under the proposed control laws \eqref{eq:ps_locked_ctrl} and \eqref{eq:ps_shape_ctrl} is established using Lyapunov theory.
\subsection{Lyapunov Stability Analysis}
For stability analysis, it is assumed that the attitude control law guaranties exponential convergence of the thrust direction $R_0\mathbf{e}_3$ to the desired force direction $u_{ct}/\|u_{ct}\|$. Under this standard time-scale separation assumption, the force projection satisfies $(u_{ct} \cdot R_0\mathbf{e}_3)R_0\mathbf{e}_3 \rightarrow u_{ct}$.
\begin{theorem}
Consider the decoupled perturbed dynamics of the Quadrotor-Robotic-Manipulator (QRM) system described by the state-space representations \eqref{eq:locked_ss}-\eqref{eq:shape_ss}. Let the filter variables $S_l$ and $S_s$ be defined as in \eqref{locked surface}-\eqref{shape surface}. 
Under the control laws \eqref{eq:ps_locked_ctrl} and
\eqref{eq:ps_shape_ctrl}, with control gain matrices $K_l$ and $K_s$ in $\mathcal{D}_n^+$ based on Assumption-\ref{gain_matrix}, and $\gamma_l$ chosen as a positive constant. Then, the following results hold:
\begin{enumerate}
    \item The tracking errors of the QRM system exponentially converge to zero as $t \to \infty$ (i.e.,$(x_1, y_1, x_2, y_2, \tau_{dL})\rightarrow(x_{1d}, y_{1d}, \dot{x}_{1d}, \dot{y}_{1d}, \hat{\tau}_{dL}) \;\; \text{as } t \to \infty$).
    \item The disturbance estimates $\hat \tau_{dL}$ and $N^\top \hat \tau_{dL}$ exponentially converge to their corresponding external disturbances $\tau_{dL}$ and $N^\top \tau_{dL}$ as $t \to \infty$.
\end{enumerate}
This stability is structurally enforced by the virtual control inputs $u_{ct}$, which dynamically generate the desired roll $\phi_r^d$ and pitch $\phi_p^d$ to compensate for the underactuated under-base.
\end{theorem}
\begin{proof}
    To investigate the stability of the  Locked subsystem \eqref{eq:locked_ss}, we propose the following positive-definite and radially unbounded candidate Lyapunov function of the filter variable $S_l$ and error disturbance $\tilde\tau_{dL}$:
    \begin{equation} \label{eq:Lyapunov_trans}
        V_l(S_l, \tilde\tau_{dL}) = \frac{1}{2}\Bigl(S_l^\top m_{qm}S_l  + \tilde\tau_{dL}^\top \tilde\tau_{dL} \Bigr).
    \end{equation}
    Taking the time derivative of \eqref{eq:Lyapunov_trans}, yields:
    \begin{equation}
\dot V_l(S_l, \tilde\tau_{dL}) = \Bigl(S_l^\top m_{qm}\dot S_l + \tilde\tau_{dL}^\top \dot {\tilde\tau}_{dL} \Bigr).
\end{equation}
Substituting \eqref{derivative locked surface} and \eqref{derivative shape surface} the derivative simplifies to:
\begin{equation}
\begin{aligned}
\dot V_l
&=
S_l^\top(m_{qm}\dot x_2 - m_{qm}\ddot x_{1d} + m_{qm}\lambda_l \dot e_{x_1}) + \tilde\tau_{dL}^\top ( \dot {\tau}_{dL} - \dot {\hat{\tau}}_{dL}), \\
& =
S_l^\top(\tau_L - G_L + \tau_{dL} - m_{qm}\ddot x_{1d} + m_{qm}\lambda_l \dot e_{x_1}) + \tilde\tau_{dL}^\top ( \dot {\tau}_{dL} - \dot {\hat{\tau}}_{dL}). \label{vl_dot}
\end{aligned}
\end{equation}
Now, expressing and substituting the locked dynamics \eqref{eq:locked_ss}, Assumption-\ref{assump_1}, estimator dynamics \eqref{eqn:dist_extimator} and \eqref{OB dynamics} in \eqref{vl_dot}, we get  
\begin{equation}
\begin{aligned}
\dot V_l
&=
S_l^\top(\tau_L - G_L + \tau_{dL} - m_{qm}\ddot x_{1d} + m_{qm}\lambda_l \dot e_{x_1}) -\tilde\tau_{dL}^\top \gamma_l {\dot {\hat\tau}}_{dL}
\end{aligned}.
\end{equation}
After substituting the control law from \eqref{eq:ps_locked_ctrl}, 
\begin{equation}
\begin{aligned}
\dot V_l
&=
S_l^\top(\tilde\tau_{dL} - K_lS_l)  -\tilde\tau_{dL}^\top \gamma_l {\tilde\tau}_{dL},\\
&= -S_l^\top K_l S_l + \tilde\tau_{dL}^\top S_l - \gamma_l{\tilde\tau}_{dL}^\top {\tilde\tau}_{dL}.
\end{aligned}
\end{equation}
Applying Young's inequality and Rayleigh-Ritz inequality \cite{Anil27072026},
\begin{equation}
   \begin{aligned}
    \dot V_l
    &\leq -\lambda_m(K_l)S_l^\top S_l + \frac{\epsilon_1}{2}\|\tilde\tau_{dL}\|^2 + \frac{1}{2\epsilon_1}\|S_l\|^2 - \gamma_l\tilde\tau_{dL}^\top {\tilde\tau}_{dL},\\
    &= -\Big(\lambda_{m}(K_l) - 1/2\epsilon_1\Big)S_l^\top S_l - (\gamma_l-\epsilon_1/2)\tilde \tau_{dL}^\top\tilde \tau_{dL},
    \end{aligned}
\end{equation}
where $\lambda_{m}(\cdot) = \lambda_{min}(\cdot)$. Choosing the control gains $K_l = (1/2\epsilon_1)I_3 + \bar K_l$ and $\gamma_l = (\epsilon_1/2) + \bar \gamma_l$, where $\bar K_l\in \mathcal{D}_3^+$, $\epsilon_1, \bar \gamma_l> 0$ and $I_3$ is the identity matrix of order 3.
With this the preceding expression will be simplified to,
\begin{equation}\label{final_vl_dot}
        \dot V_l = -\lambda_m(\bar K_l)S_l^\top S_l - \bar\gamma_l\tilde\tau_{dL}^\top \tilde\tau_{dL}.
\end{equation}
Next for the shape subsystem state-space dynamics, consider the candidate Lyapunov function using the filter variable \eqref{shape surface} as:
\begin{equation}
V_s(S_s) = \frac{1}{2}S_s^\top M_S S_s.
\label{eq:lyapunov_shape}
\end{equation}
Taking time derivative,
\begin{equation}
\dot V_s(S_s) = S_s^\top M_S \dot S_s + S_s^\top \dot M_S  S_s.
\end{equation}
Substituting \eqref{derivative shape surface} and using the state-space dynamics \eqref{eq:shape_ss}
the derivative simplifies to:
\begin{equation}
\begin{aligned}
\dot V_s=S_s^\top(\tau_S - C_s(y_1,y_2)y_2 + N^\top \tau_{dL} - M_S\ddot y_{1d} + M_S\lambda_s\dot e_{y_1}) \\ +S_s^\top \dot M_S  S_s.
\end{aligned}
\end{equation}
By substituting \eqref{error states},\eqref{shape surface},\eqref{eq:ps_shape_ctrl}, and exploiting the property of Euler-Lagrange systems that $\dot M_s - 2C_s$ is skew-symmetric, the Lyapunov derivative reduces to,
\begin{equation} \label{eq:deri_lya_sha}
\dot V_s = S_s^\top(N^\top\tilde\tau_{dL} - K_sS_s).
\end{equation}
This expression can be re-written by using Young's and Rayleigh-Ritz inequality, 
\begin{equation}
     \dot V_s  \leq  - \lambda_m(K_s)S_s^\top S_s  + \frac{\epsilon_2}{2}||S_s||^2 + \frac{1}{2\epsilon_2}  \|N^\top \tilde\tau_{dL}\|.  
\end{equation}
Referring definition of spectral norm, which says $\|N^\top\tilde\tau_{dL}\|^2 \leq  \|N\|^2 \|\tilde{\tau}_{dL}\|^2$, this gives:
\begin{equation}
    \dot V_s \leq -\Big(\lambda_m(K_s) - \frac{\epsilon_2}{2}\Big)S_s^\top S_s +\frac{1}{2\epsilon_2} \|N\|^2 \|\tilde{\tau}_{dL}\|^2.
\end{equation}
Selecting  $K_s = (\epsilon_2/2)I_5 + \bar K_s$, where $\bar K_s\in \mathcal{D}_5^+$, and $I_5$ is the identity matrix of order 5,
\begin{equation} \label{final_Vs_dot}
    \dot V_s = -\lambda_m(\bar K_s)S_s^\top S_s +\frac{1}{2\epsilon_2} \|N\|^2 \|\tilde{\tau}_{dL}\|^2.
\end{equation}
To analyze the stability of the overall closed-loop system
\eqref{eq:locked_ss}-\eqref{eq:shape_ss}, consider the composite
Lyapunov function candidate,
\begin{equation}
V(S_l,S_s,\tilde{\tau}_{dL})
=
V_l(S_l,\tilde{\tau}_{dL}) + V_s(S_s),
\label{eq:lyapunov}
\end{equation}
where $V_l$ and $V_s$ are the positive-definite Lyapunov functions
associated with the locked and shape subsystems, respectively.
Since both $V_l$ and $V_s$ are positive definite and radially unbounded,
the composite function $V$ is also positive definite and radially
unbounded.
Taking time derivative with respect to time and substituting \eqref{final_vl_dot},\eqref{final_Vs_dot}, yields:
\begin{equation}
\dot V = -\lambda_m(\bar K_l)S_l^\top S_l -\lambda_m(\bar K_s)S_s^\top S_s -\frac{1}{2\epsilon_2}\Big(\bar\gamma_l- \|N\|^2\Big) \|\tilde{\tau}_{dL}\|^2.
\label{final_lya_deri}
\end{equation}
 Based on Claim \ref{Norm of N}, which establishes the upper bound for the norm of N as $\lVert N\rVert^2 \leq \frac{5m_1^2l^2}{4m_{qm}^2}$, the choosing $\bar\gamma_l = \|N\|^2 + 2\epsilon_2\gamma$, where $\epsilon_2,\gamma >0$. Here $I_3$, $I_5$ are the identity matrices of order $3$ and $5$, respectively. After substituting these gains in \eqref{final_lya_deri}, this yields,
\begin{align}
    \dot V \leq - \Gamma V \leq 0,
\end{align}
where $\Gamma = 2 \min(\lambda_m(\bar K_l), \lambda_m(\bar K_s),\gamma)$.
This shows that surfaces $S_l$, $S_S$, and $\tilde \tau_{dL}$ converge to zero exponentially.
This implies $\dot S_l \rightarrow 0$ and $\dot S_s \rightarrow 0$
We have $S_l = e_{x_2} + \lambda_le_{x_1}$; this can be further re-written as \[\dot e_{x_1} = S_l -  \lambda_le_{x_1}.\] Now multiply this by $e^{\lambda_lt}$; this gives,
\begin{equation*}
    \begin{aligned}
       e^{\lambda_lt} \dot e_{x_1} &= e^{\lambda_lt}S_l -  e^{\lambda_lt}\lambda_le_{x_1},\\
    e^{\lambda_lt}S_l    &=e^{\lambda_lt} \dot e_{x_1}+e^{\lambda_lt}\lambda_le_{x_1},\\
    \frac{d}{dt}( e^{\lambda_lt} e_{x_1}) &=e^{\lambda_lt}S_l.
    \end{aligned}
\end{equation*}
Integrating both sides from time 0 to t,
\begin{equation*}
    \begin{aligned}
        e^{\lambda_lt} e_{x_1}(t) - e_{x_1}(0) &= \int_0^t e^{\lambda_lt}S_l dt,\\
         e_{x_1}(t) &=  e^{-\lambda_lt}e_{x_1}(0)+
         e^{-\lambda_lt}\int_0^t e^{\lambda_lt}S_l dt.
    \end{aligned}
\end{equation*}
as time goes to infinity, $e_{x_1}$ will converge to zero. Using $e_{x_1} \rightarrow 0$ and $S_l \rightarrow 0$ in \eqref{locked surface}, this shows that $e_{x_2}$ converges to zero. Similarly, one can show that $e_{y_1}$ and $e_{y_2}$ converge to zero.
This proves that the system dynamics is exponentially stable and the position and velocity states converge to their desired values exponentially, while disturbances $\tau_{dL}$ and $N^\top \tau_{dL}$ converge to the estimated disturbances $\hat \tau_{dL}$ and $N^\top \hat \tau_{dL}$, respectively as $t \to \infty$.
\end{proof}
\section{Simulation Results and analysis} \label{sec: simulation}
This section presents the numerical simulation results to validate the effectiveness and robustness of the proposed control framework for the Quadrotor-Robotic-Manipulator (QRM) system. To evaluate performance under realistic operating conditions, simulations are performed using the full, nonlinear, decoupled state-space perturbed dynamics \eqref{eq:locked_ss}-\eqref{eq:shape_ss} derived in Section~\ref{Control design}. The physical configuration of the simulated platform matches the description in Section~\ref{System modelling} and is illustrated in Fig.~\ref{fig:QRMS}, with specific physical parameters such as component masses, moments of inertia, and geometric dimensions summarized in Table~\ref{tab:sim_params}.
\begin{table}[t]
\caption{Simulation Parameters of the Quadrotor–Manipulator System}
\label{tab:sim_params}
\centering
\begin{tabular}{lll}
\hline
\textbf{Parameter} & \textbf{Symbol} & \textbf{Value} \\
\hline
Quadrotor mass & $m_0$ & 5 kg \\
Manipulator mass & $m_1$ & 1 kg \\
Total system mass & $m_{qm}$ & 6 kg \\
Quadrotor inertia & $\mathbb{I}_0$ & diag$(0.2,0.2,0.2)$ kg·m$^2$ \\
Link inertia & $\mathbb{I}_1$ & diag$(0.002,0.002.0.002)$ kg·m$^2$ \\
Link length & $l$ & 0.5 m \\
Gravitational acceleration & $g$ & 9.81 m/s$^2$ \\
Locked control gain & $\bar K_l$ & diag$(2.5,2.6,2.6)$ \\
Shape control gain & $\bar K_s$ & diag$(5.7,5.6,5.6, 5.8, 5.9)$ \\
Disturbance Gain & $\gamma_l$ & 4 \\
Locked position gain & $\lambda_l$ & diag$(0.5,0.6,0.6)$\\
Shape position gain & $\lambda_s$ & diag$(0.5,0.6,0.6,0.4,0.8)$\\
\hline
\end{tabular}
\end{table}
The control architecture exploited in this validation leverages the structural decoupling between the locked and shape dynamics to enable independent regulation of the total system $CoM$ motion via a geometric controller \eqref{eq:ps_locked_ctrl}, while simultaneously stabilizing the attitude dynamics and RM joint angles through the shape controller \eqref{eq:ps_shape_ctrl}. 
Although the proposed control framework does not directly regulate the end-effector position, it achieves this indirectly through the coordinated control of the system’s $CoM$, quadrotor attitude, and RM joint angles. Specifically, the end-effector's physical location is structurally coupled to these controlled states. For an end-effector vector defined relative to the final link frame as $r_{ee} = [0; 0; -l]^T$, its absolute position in the inertial frame ($\eta_{ee} \in \mathcal{F}_I$) is explicitly mapped via:
\begin{equation}
\eta_{ee} = \eta_0 + R_o R_1^b r_{ee},
\end{equation}
where $\eta_0$ denotes the quadrotor position, and $R_o, R_1^b$ represents the respective rotation matrices.
To evaluate the proposed controller across diverse operational profiles, the system is tested on both a point-stabilization task and a trajectory-tracking task designed for multi-waypoint pick-and-place operations. For the latter, a smooth reference trajectory is required for the CoM of the system, denoted as $\eta^d_{qm}$, which must be synthesized from a desired end-effector target $\eta^d_{ee}$. This trajectory tracking is uniquely determined through the geometric allocation strategy detailed in Algorithm~\ref{alg:task_execution}, which yields:
\begin{equation}
\eta^d_{qm} = \eta^d_{ee} + \xi(m_0, m_1, \zeta) r_{ee}.
\end{equation}
By utilizing this task-space execution strategy to link the physical end-effector space with the virtual system coordinates, a series of comparative simulation scenarios is conducted to demonstrate the tracking precision and transient performance of the proposed framework. Importantly, the trajectory is designed to preserve an adequate safety margin between the ceiling and the target location, ensuring sufficient operational space for precise pick-and-place maneuvers.
\begin{algorithm}[t]
\caption{Desired $CoM$ of QRM system  for Desired End-Effector Position}
\label{alg:task_execution}
\begin{algorithmic}[1]
\REQUIRE Desired end-effector position 
         $\eta^d_{ee} = [a_0,\, b_0,\, c_0]^\top$
\REQUIRE Quadrotor mass $m_0$, manipulator mass $m_1$.
\ENSURE Desired quadrotor $CoM$ position $\eta^d_{0}$

\STATE Quadrotor desired attitude $R^d_{0}$
\STATE The RM desired rotation ${R^b_{1}}^{d}$
\STATE Position vector of the End-effector relative to the $\{\mathcal{F}^b_M\}$: $r_{ee} = [0\;\;0\;\;-l]^\top$
\STATE Position vector of the $CoM$ of RM relative to the $\{\mathcal{F}^b_M\}$: $r_{M} = [0\;\;0\;\;-l/2]^\top$
\STATE Compute QRM system $CoM$ due to quadrotor and RM 
\[\eta_{qm} =  \eta_{0} + R_0(\Phi) \eta^b_1 \left(\frac{m_1}{m_0 + m_1}\right)\]
\STATE Compute end effector position in inertial frame
\[\eta_{ee} = \eta_{0} + R_0(\Phi)R_z(\theta_1)R_y(\theta_2)r_{ee}\]
\STATE Compute desired QRM system position
       \[
       \eta^d_{qm} \gets \eta^d_{ee} + \xi(m_0,m_1,\zeta)r_{ee}
       \]
\RETURN $\eta_{qm_d}$
\end{algorithmic}
\end{algorithm}
\begin{figure}
    \centering
    \includegraphics[width=1\linewidth]{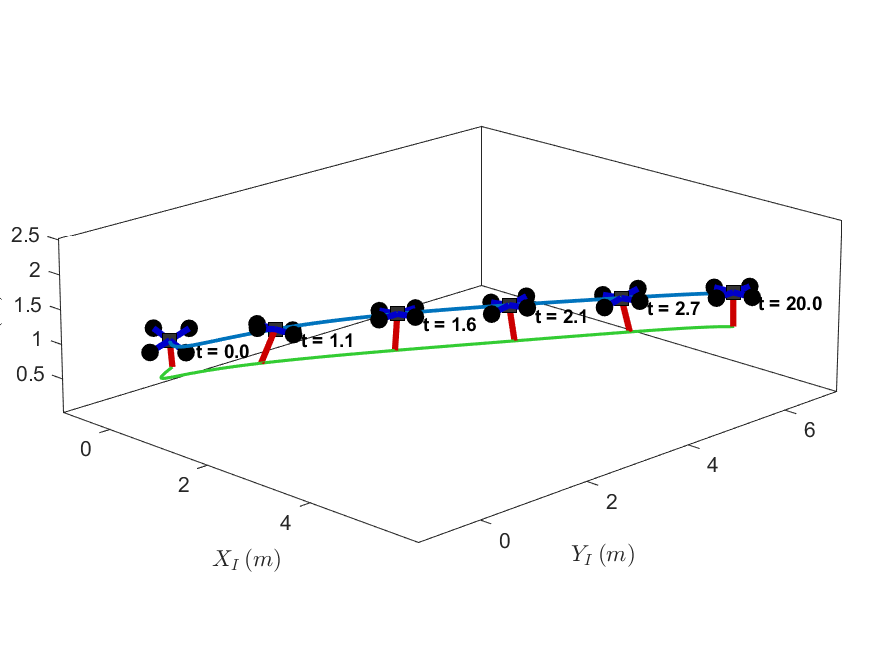}
    \caption{Snapshots of point stabilization results showing end-effector and quadrotor trajectories; the blue curve represents the quadrotor position $\eta_0$ while the green curve is for the RM end-effector position $\eta_{ee}$.}
    \label{animation_PS}
\end{figure}
\begin{figure}
    \centering
    \includegraphics[width=1\linewidth]{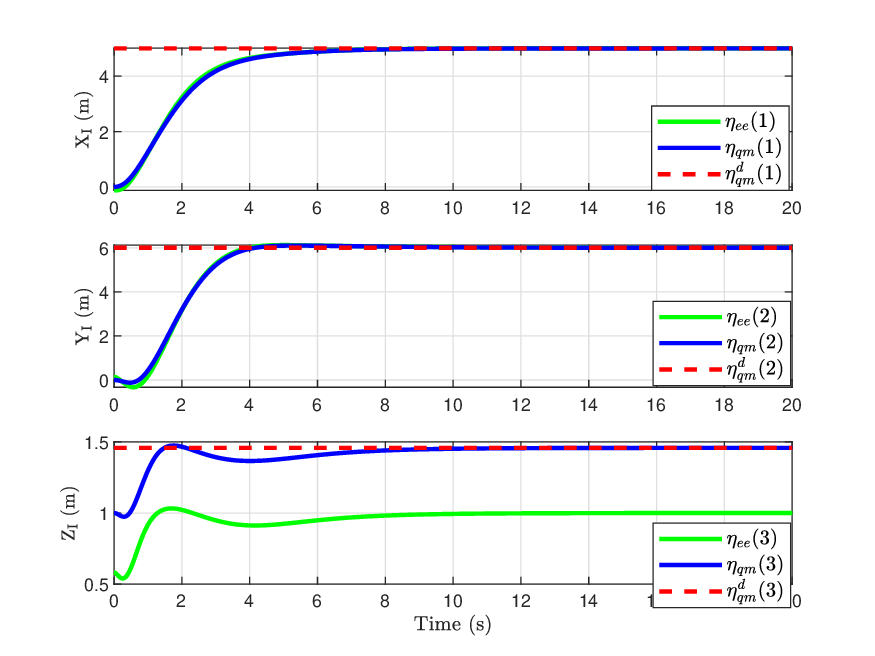}
    \caption{Position plots  of $\eta_{qm}$,  $\eta_{ee}$ along with $\eta^d_{qm}$.}
    \label{qm_ee_effe_pos}
\end{figure}
\begin{figure}
\centering
\includegraphics[width=1\linewidth]{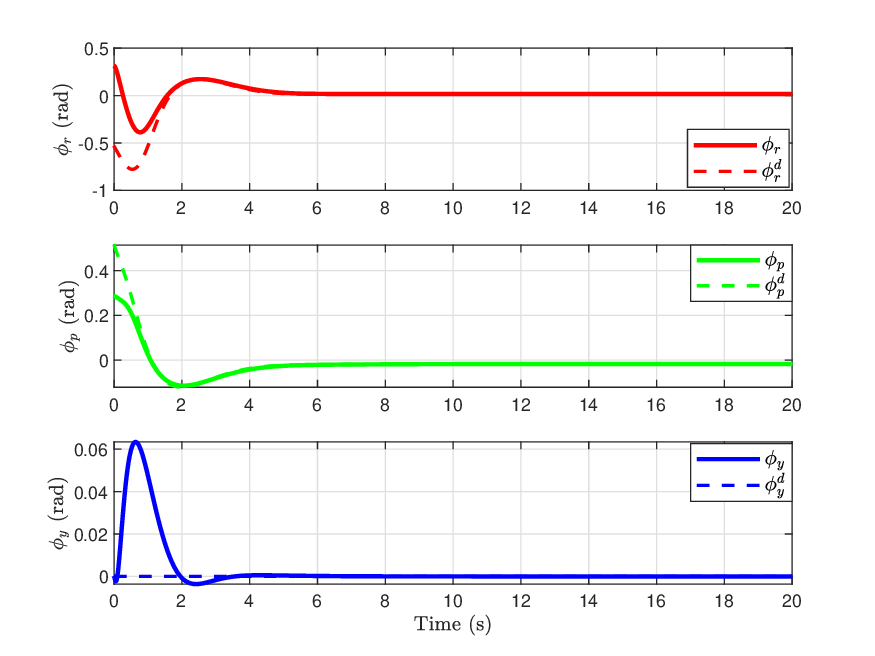}
\caption{Roll, pitch, and yaw response of the quadrotor during the end-effector point stabilization task.}
\label{QM_orien}
\end{figure}
\begin{figure}
\centering
\textbf{}\includegraphics[width=1\linewidth]{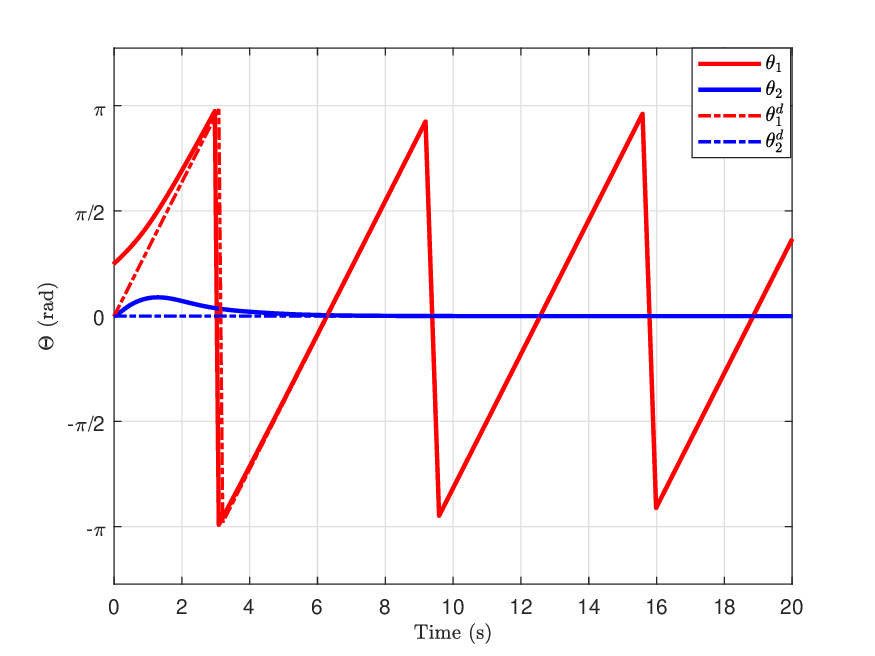}
\caption{Manipulator joint trajectories during the end-effector point stabilization task: (a) joint angle $\theta_1$ and (b) joint angle $\theta_2$.}
\label{manipulator_angles}
\end{figure}
\begin{figure}
    \centering
\includegraphics[width=1\linewidth]{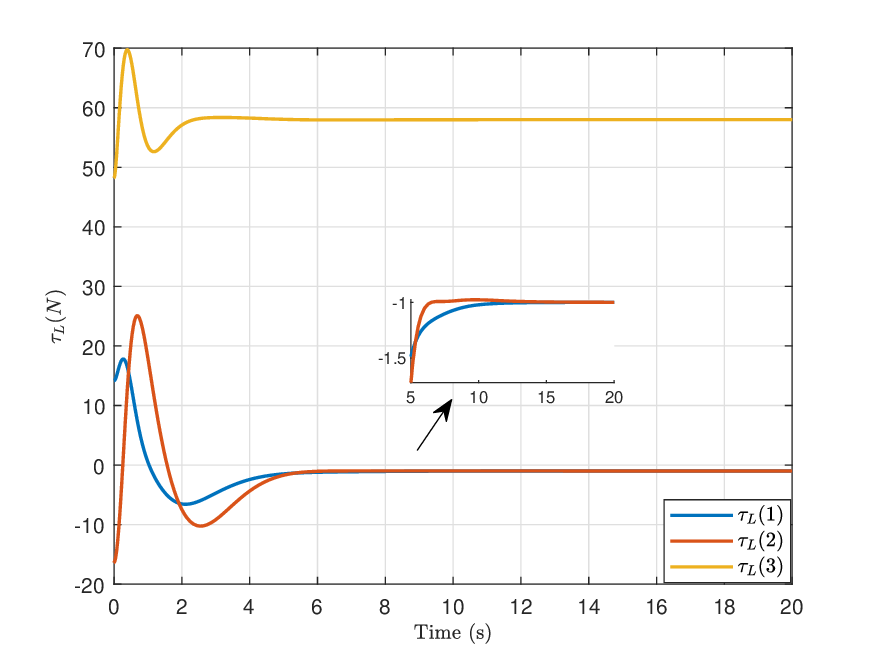}
    \caption{Locked control $\tau_L$ for end-effector point stabilization task}
    \label{Locked_con_ps}
\end{figure}
\begin{figure}
    \centering
\includegraphics[width=1\linewidth]{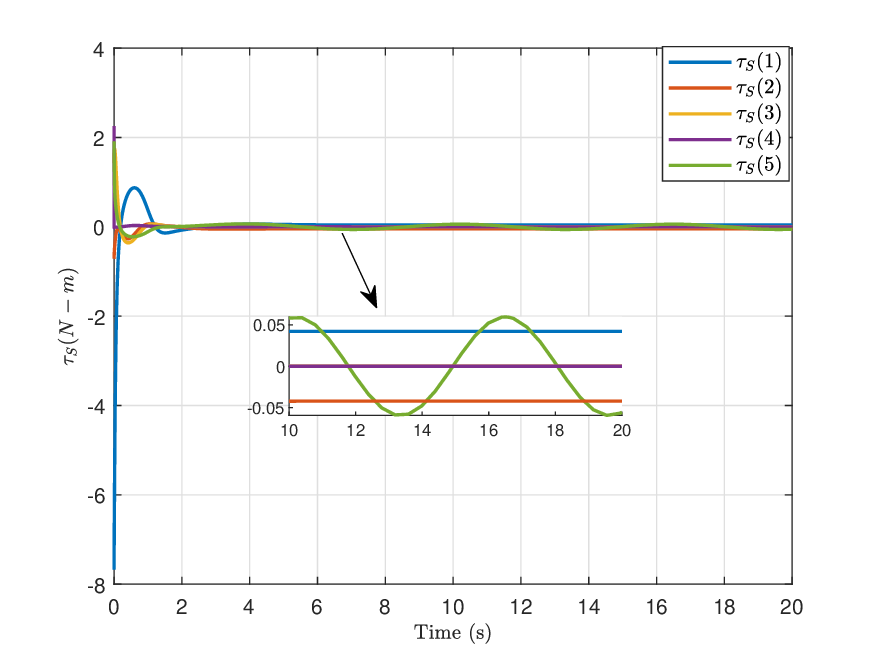}
    \caption{Shape control $\tau_S$ for end-effector Point stabilization task.}
    \label{shape_con_ps}
\end{figure}
\begin{figure}
    \centering
    \includegraphics[width=1\linewidth]{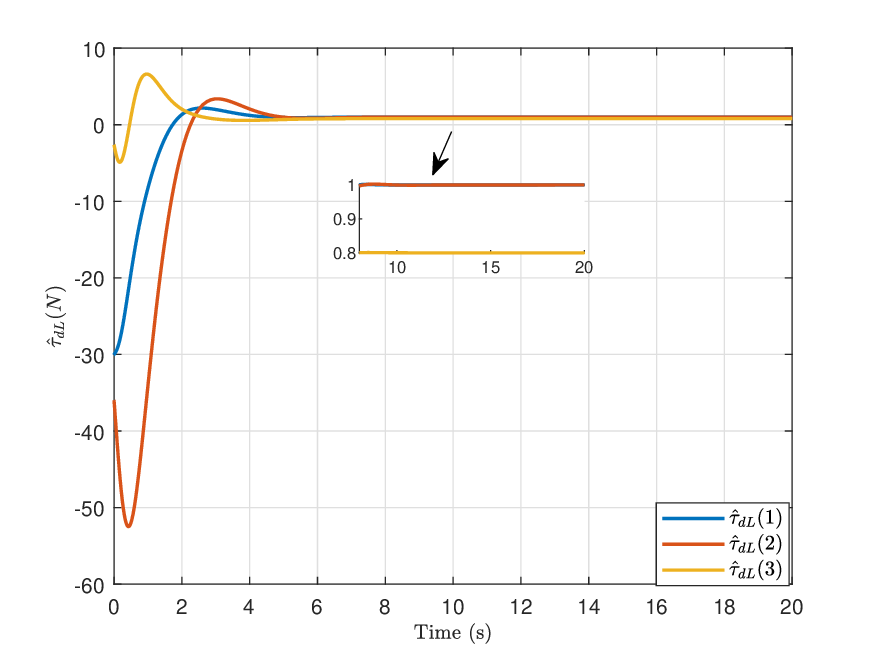}
    \caption{Estimated disturbance $\hat\tau_{dL}$ on locked system for end-effector point stabilization task.}
    \label{ob_ps_taudL_hat}
\end{figure}
\begin{figure}
    \centering
    \includegraphics[width=1\linewidth]{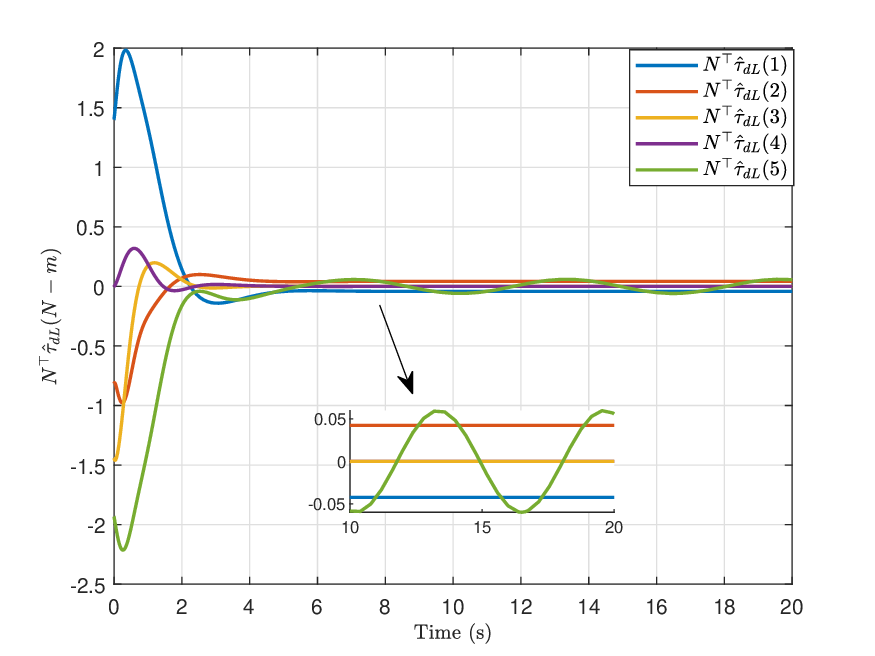}
    \caption{Estimated disturbance $N^\top \hat \tau_{dL}$ on shape system during Point stabilization task.}
    \label{ob_ps_N'tau_dL_tilde}
\end{figure}
\subsection{End-Effector Point Stabilization}
In the first task, the objective is to stabilize the end-effector at a fixed point in the inertial frame. The desired end-effector position $\eta_{ee}^d$ is selected as a constant setpoint. The QRM system is initialized away from equilibrium, and convergence of the end-effector position is achieved through the coordinated motion of the system's $CoM$, attitude, and RM joints.
The desired end-effector position is selected as
\[
\eta_{ee}^d =
\begin{bmatrix}
5 \;\; 6 \;\; 1
\end{bmatrix}^\top \text{ m},\] so from algorithm~\ref{alg:task_execution}, $\eta^d_{qm}$ is calculated, 
desired internal variables 
\[\Theta^d = [wt \;\; 0], \quad \Phi^d = [\phi_r^d \;\; \phi_p^d \;\; 0] rad,\] where $\phi_r^d$, $\phi_p^d$ are given by the equation \eqref{eq:roll desired}--\eqref{eq:pitch desired} and $w = 1$.
The system is initialized away from equilibrium with
\[\eta_{qm}(0) =
\begin{bmatrix}
0\;\; 0 \;\; 1
\end{bmatrix}^\top \text{m},
\quad
\Phi(0) = \left[\pi/10\;\; \pi/11\;\; 0\right]^\top rad,\]
\[\Theta(0) = \left[\pi/4\;\; 0\right]^\top rad, \quad \tau_{dL} = [1\;\;1\;\;0.8]^\top N.\]
Fig.~\ref{qm_ee_effe_pos} illustrates the end-effector trajectory during the point convergence of the end-effector, alongside the time histories of the system $CoM$ position, $\eta_{qm}$, compared against the desired reference trajectories, $\eta^d_{qm}$, for the $x$, $y$, and $z$ axes in $\mathcal{F}_I$. The close agreement between the actual and desired states confirms high-fidelity trajectory tracking. Notably, because the desired RM joint angles $\theta_2$ is set to zero, the RM remains vertically aligned at $90^\circ$ relative to the quadrotor’s body frame. Consequently, the $x$ and $y$ coordinates of the end-effector coincide with the horizontal components of $\eta_{qm}$. However, the end-effector's $z$-position remains lower than the CoM altitude because the manipulator is mounted beneath the quadrotor's airframe, introducing a vertical offset in the $Z$ direction.
The attitude response of the quadrotor during the maneuver is shown in Fig.~\ref{QM_orien}; roll, pitch, and yaw angles remain bounded and converge smoothly to their steady-state values within 4 s, ensuring proper thrust alignment throughout the stabilization task.
Furthermore, the evolution of the RM joint angles shown in Fig.~\ref{manipulator_angles} confirms that the joints converge to their desired position approximately at 6 s and remain steady. 
Furthermore, Fig. \ref{Locked_con_ps} and Fig. \ref{shape_con_ps} illustrate the performance of the locked (thrust) and shape (torque) controllers, respectively. These results correlate directly with the external disturbances shown in Figs. \ref{ob_ps_taudL_hat} and \ref{ob_ps_N'tau_dL_tilde}, confirming that the control torques effectively compensate for external perturbations. These results demonstrate control effectiveness during point stabilization relative to Earth's gravity, as indicated in the locked control plot. 
This demonstrates that the arm can be controlled precisely without causing the quadrotor to wobble or oscillate during the process.
\subsection{Simulation of Pick-and-Place Operation}
\label{sec:simulation}
This subsection presents a numerical simulation study to validate the effectiveness of the proposed waypoint-based trajectory tracking framework for a QRM system performing a pick-and-place operation. The objective is to demonstrate smooth and accurate navigation between predefined locations using a dynamically feasible reference trajectory. For such operations, the reference trajectory must be sufficiently smooth to avoid aggressive control inputs and actuator saturation; specifically, the continuity of position, velocity, and acceleration is required, while higher-order smoothness is desirable for system robustness.
Let ${\eta^d_{qm}}(t) \in \mathbb{R}^3$ denote the desired Cartesian position of the $CoM$ of QRM system. The mission is specified by a sequence of waypoints $\mathcal{WP} = \{ \mathbf{p}_0, \mathbf{p}_1, \ldots, \mathbf{p}_5 \}, \;\; \mathbf{p}_i \in \mathbb{R}^3$ with corresponding arrival times $\mathcal{T} = \{ t_0, t_1, \ldots, t_5 \}$. Consider $i = 0$ to $i = 5$, the trajectory between consecutive waypoints is generated independently over each time interval $[t_i, t_{i+1}]$ using a minimum-jerk optimal control criterion. Among all admissible trajectories connecting $\mathbf{p}_i$ and $\mathbf{p}_{i+1}$ over $[t_i, t_{i+1}]$, we select the one that minimizes the squared jerk:
\begin{equation} \label{eq:minjerk}
\min_{(t)} 
\int_{t_i}^{t_{i+1}} 
\left\| {\eta^d_{qm}}^{(3)}(t) \right\|^2 dt,
\end{equation}
where ${\eta^d_{qm}}^{(3)}(t)$ is the third order derivative of the desired $CoM$ $\eta_{qm}^d$
, subject to the boundary conditions
\begin{equation}\label{boundary conditions}
\left.
\begin{aligned}
\eta_{qm}^d(t_i) &= p_i, \quad \eta_{qm}^d(t_{i+1}) = p_{i+1}, \\
\dot{\eta}_{qm}^d(t_i) &= 0, \quad \dot{\eta}_{qm}^d(t_{i+1}) = 0, \\
\ddot{\eta}_{qm}^d(t_i) &= 0, \quad \ddot{\eta}_{qm}^d(t_{i+1}) = 0.
\end{aligned}
\right\}
\end{equation}
This cost function penalizes rapid changes in acceleration and results in smooth motion profiles that are well suited for quadrotor systems.\\
From the calculus of variations, using the Euler lagrange equation, the solution to this cost function \eqref{eq:minjerk} is a fifth-order polynomial in time, which is expressed as ${\eta^d_{qm}}(t) = \sum_{k=0}^{5} a_k t^k$. By using the boundary conditions given by \eqref{boundary conditions}, the optimal solution is written as ${\eta^d_{qm}}(t) = {p}_i + \left( {p}_{i+1} - {p}_i \right) s( \varsigma )$, where $s( \varsigma ) = 10 \varsigma ^3 - 15 \varsigma ^4 + 6 \varsigma ^5$ is the minimum-jerk blending function for a rest-to-rest objective, that is, fixed initial and final positions, and the normalized time $\varsigma$ = $\frac{t - t_i}{t_{i+1} - t_i}$. This function satisfies the six boundary conditions at the interval limits, ensuring that the trajectory and its first two derivatives remain continuous across all waypoint transitions.
The corresponding velocity, acceleration, and higher-order derivatives are derived by differentiating ${\eta^d_{qm}}(t)$ with respect to time using the chain rule $\frac{d}{dt} = \frac{1}{t_{i+1}-t_i} \frac{d}{d \varsigma }$, yielding:
\begin{equation}
\dot \eta^d_{qm}(t) = 
\frac{{p}_{i+1} - {p}_i}{t_{i+1}-t_i} \, \dot{s}( \varsigma ),
\end{equation}
doing recursively, the higher-order derivatives are computed as,
\begin{equation}\label{desired derivatives}
\left.
\begin{aligned}
\ddot\eta^d_{qm}(t) &=
\frac{{p}_{i+1} - {p}_i}{(t_{i+1}-t_i)^2} \, \ddot{s}( \varsigma ), \\
{\eta^d_{qm}}^{(3)}(t) &=
\frac{{p}_{i+1} - {p}_i}{(t_{i+1}-t_i)^3} \, {s}^{(3)}( \varsigma ), \\
{\eta^d_{qm}}^{(4)}(t) &=
\frac{{p}_{i+1} - {p}_i}{(t_{i+1}-t_i)^4} \, {s}^{(4)}( \varsigma ).
\end{aligned}
\right\}
\end{equation}
The derivatives of the blending function or path parameterization are explicitly given by
\begin{equation}\label{path parameterization derivatives}
\left.
\begin{aligned}
\dot{s}( \varsigma ) &= 30 \varsigma ^2 - 60 \varsigma ^3 + 30 \varsigma ^4, \\
\ddot{s}( \varsigma ) &= 60 \varsigma  - 180 \varsigma ^2 + 120 \varsigma ^3, \\
{s}^{(3)}( \varsigma ) &= 60 - 360 \varsigma  + 360 \varsigma ^2, \\
{s}^{(4)}( \varsigma ) &= -360 + 720 \varsigma
\end{aligned}
\right\}.
\end{equation}
The reference signals $\eta^d_{qm}$, \eqref{desired derivatives}, and \eqref{path parameterization derivatives} are used directly for feedforward and feedback control of the quadrotor.
With this trajectory generation setup, simulations will be performed and validated for pick and place payload delivery and transportation applications. The QRM system is initialized with
$\eta_{qm}(0) =
\begin{bmatrix}
2 & 1 & 2
\end{bmatrix}^\top \text{m}$,
\;\;
$\Phi(0) = \left[0\;\; 0\;\; {\pi}/{6}\right]^\top rad$,
$
\Theta(0) = [5\pi/6 \;\; {\pi}/{3}]^\top rad $
while all initial internal configuration velocities $\dot \zeta = [\dot \Phi \;\; \dot \Theta]^\top$ are set to zero. The system is subjected to constant external disturbance $\tau_{dL} = [2 \;\; 2 \;\;1.8]^\top N$.
The flight profile consists of sequential way-point tracking across the take-off, pick, transport, drop, and landing phases. The trajectory coordinates are defined by \([0\;\;0\;\;1]^\top\), \([0\;\;0\;\;3]^\top\), \([5\;\;0\;\;1]^\top\), and a final hovering (landing) position of \( [3\;\;3\;\;1]^\top \text{m}\).
Each segment is generated independently using the proposed minimum-jerk formulation, ensuring smooth transitions and dynamically feasible motion throughout the mission.
\begin{figure}[!htbp]
    \centering
\includegraphics[width=1\linewidth]{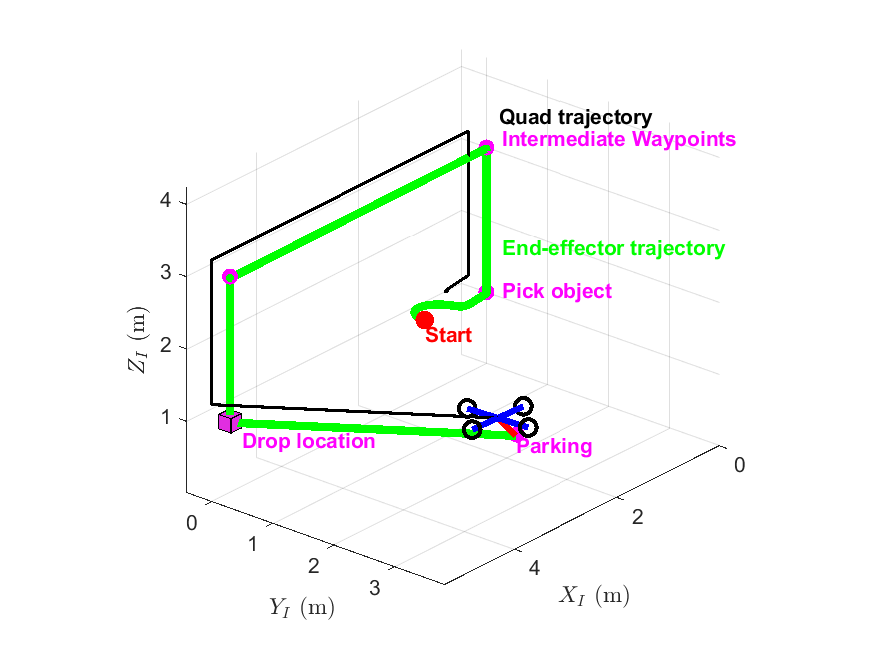}
    \caption{Pick-and-place operation in 3D space, where magenta markers (balls) denote the reference way-points, red marker is the starting/initial position, and square-box represent object, while solid black and solid green speaks for quadrotor path and end end-effector path respectively.}
    \label{pick and place in 3d}
\end{figure}
\begin{figure}[!htbp]
\centering
\includegraphics[width=1\linewidth]{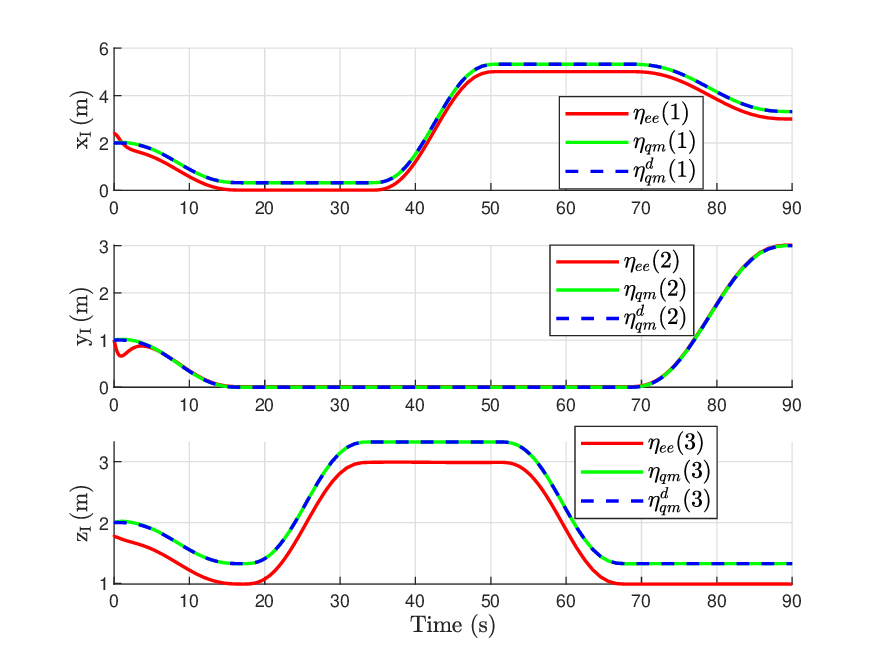} 
\caption{End effector and QRM CoM position in the inertial frame during the pick and place way-point tracking operation.}
\label{end_eff_pos_wp}
\end{figure}
\begin{figure}[!htbp]
\centering
\includegraphics[width=1\linewidth]{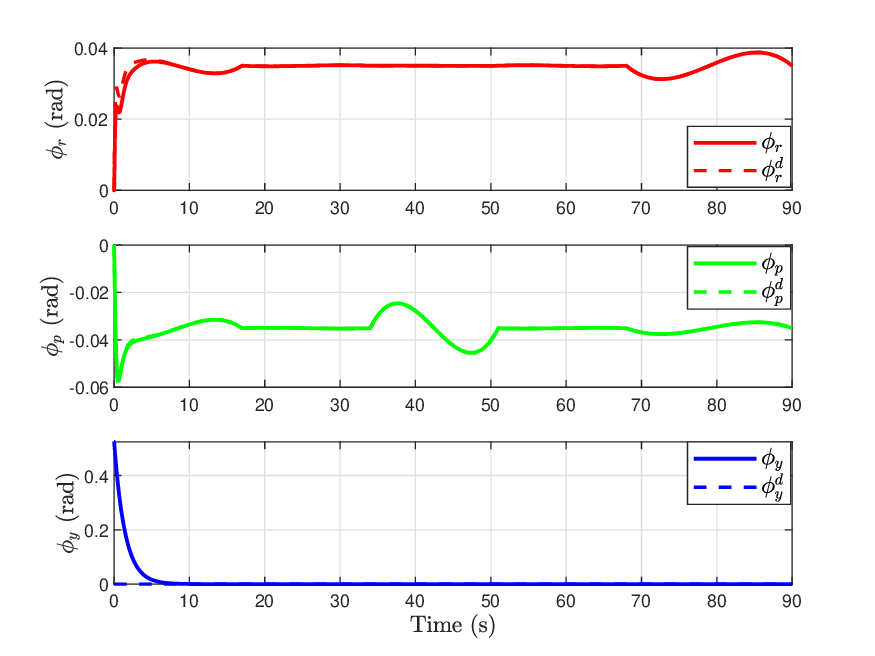}
\caption{Euler angles response of the quadrotor during the pick and place way-point tracking operation.}
\label{fig:QM_orien_TT}
\end{figure}
\begin{figure}[!htbp]
    \centering
\includegraphics[width=1\linewidth]{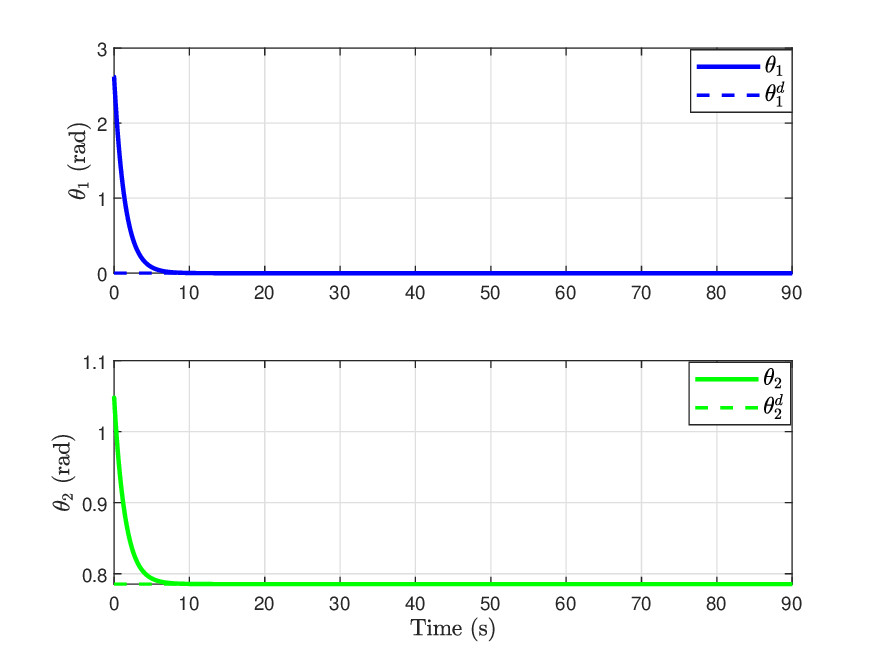}
    \caption{RM joint angles follows the desired set-points during the pick and place way-point tracking operation.}
    \label{ob_pp_JA}
\end{figure}
\begin{figure}[!htbp]
    \centering
    \includegraphics[width=1\linewidth]{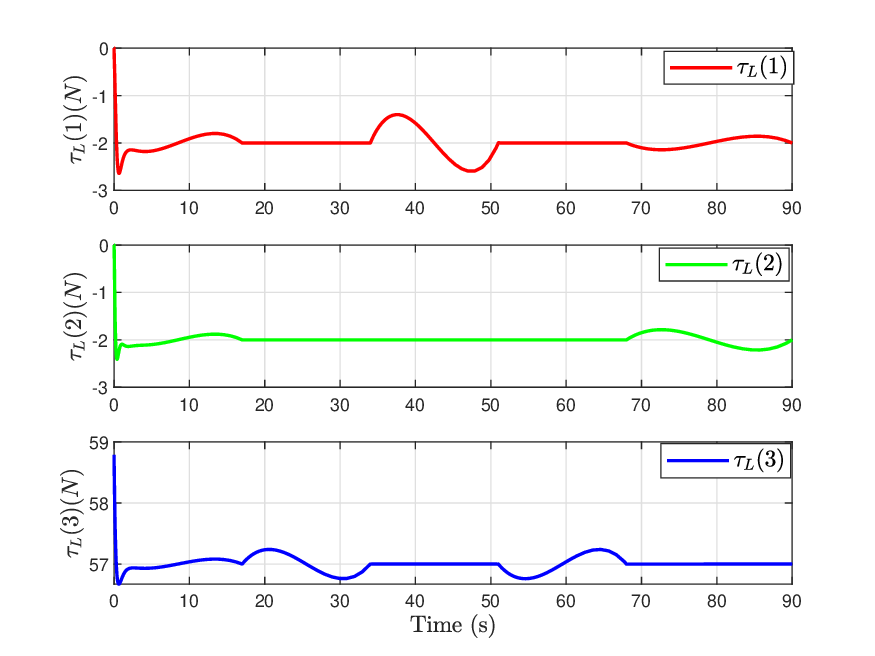}
    \caption{Locked control or thrust vector of the quadrotor during the pick and place task.}
    \label{Locked_control_PP}
\end{figure}
\begin{figure}[h]
    \centering
    \includegraphics[width=1\linewidth]{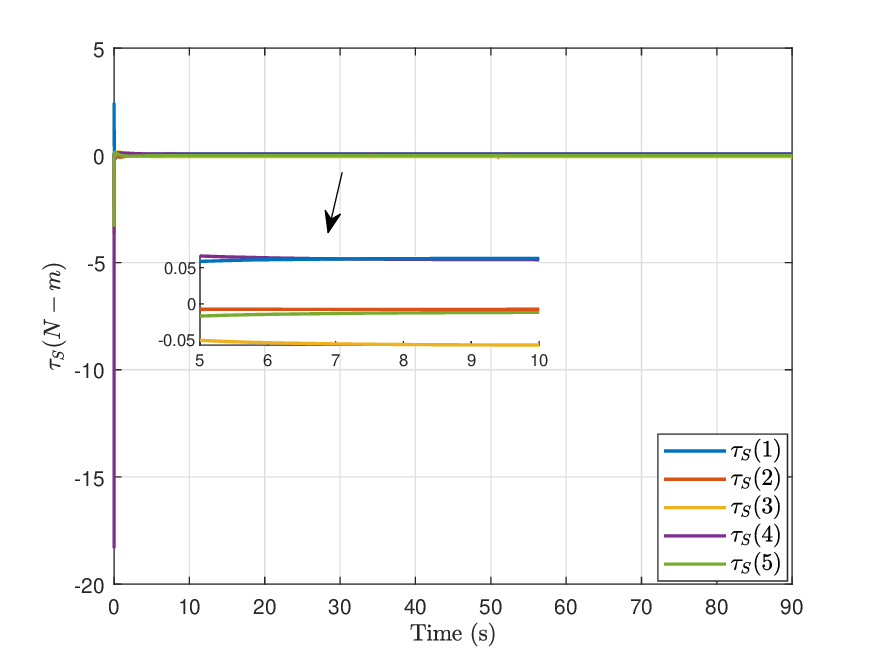}
    \caption{Shape control(combination of $\tau_L,u_\Phi,u_\theta$) during the pick and place way-point tracking operation.}
    \label{Shape_control_PP}
\end{figure} 
The three-dimensional flight path of the QRM system for the pick and place task is illustrated in Fig.~\ref{pick and place in 3d}, where the RM picks the object from $[0\; 0 \; 1]^\top$ and drops it to $[5 ~ \; 0 ~ \; 1]^\top$. 
Fig~\ref{end_eff_pos_wp} represents the results of the end-effector trajectory during the pick-and-place task, along with the time histories of the system $CoM$ position $\eta_{qm}$ components compared with the desired reference trajectories  $\eta^d_{qm}$ along the $x$, $y$, and $z$ axes, respectively. 
The close agreement between the actual and desired positions confirms accurate trajectory tracking. 
Fig. \ref{fig:QM_orien_TT} depicts the smooth transition of the attitude variable $\Phi$ of the QRM system for the pick and place task. It is noted that the end effector reaches to pick the object from the initial condition with a small initial transient. It then follows a straight line with trajectory generation. It is evident that when the system goes from $[0 \;\; 0 \;\; 3]^\top$ to $[5 \;\; 0 \;\; 3]^\top$, there is only pitch while maintaining a constant roll ($\phi_r$). Lastly, both $\phi_r$ and $\phi_p$ are involved while yaw is zero.  
Moreover, Fig.~\ref{ob_pp_JA} confirms that the joints converge to their desired position smoothly and remain steady.
Lastly, Fig.\ref{Locked_control_PP} and Fig.\ref{Shape_control_PP} display the locked and shape control inputs, respectively. As shown, these signals remain continuous throughout the operation and do not converge to zero because of the disturbance; these controller values can be justified from the estimated disturbance plots Fig.\ref{ob_pp_taudL_hat} and Fig.\ref{ob_pp_NTtaudL_hat}. Notably, a significant increase in the locked control input is observed during motion along the x and y axes, reflecting the actuation effort required for lateral translation. 
\begin{figure}[!t]
    \centering
    \includegraphics[width=1\linewidth]{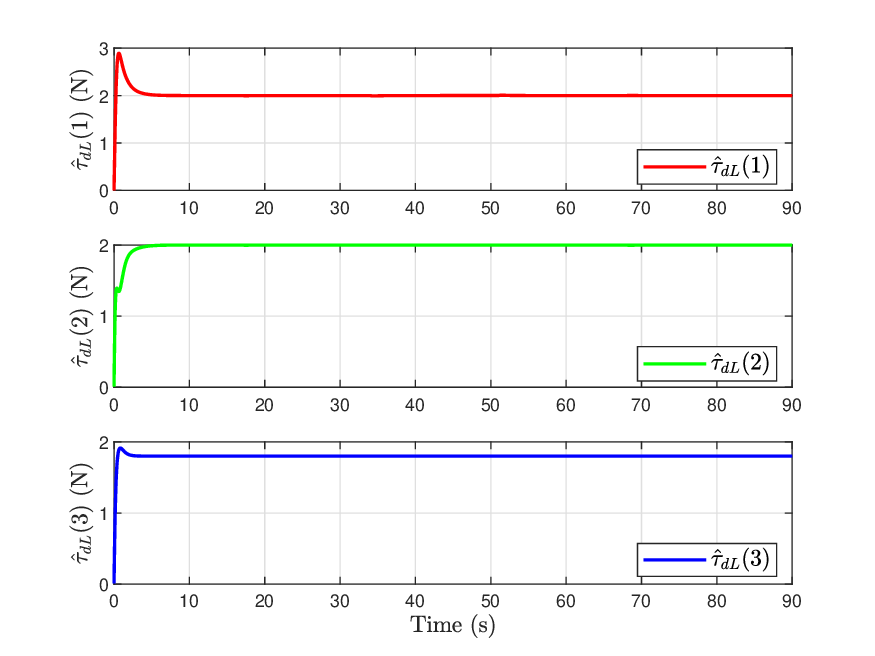}
    \caption{Estimated disturbance on locked system during the pick and place way-point tracking operation.}
    \label{ob_pp_taudL_hat}
\end{figure} 
\begin{figure}[!t]
    \centering
    \includegraphics[width=1\linewidth]{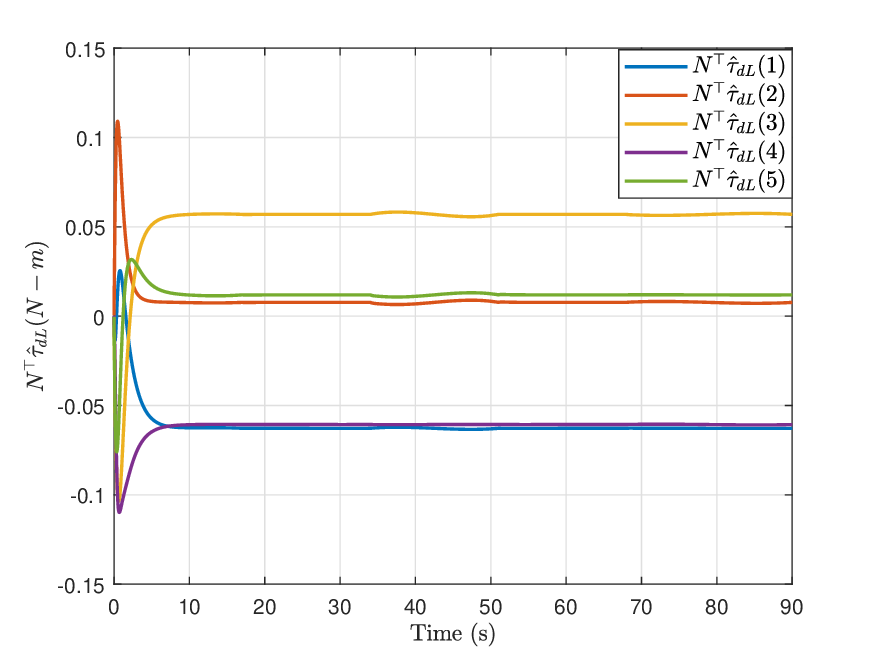}
    \caption{Estimated disturbance on shape system during the pick and place way-point tracking operation.}
    \label{ob_pp_NTtaudL_hat}
\end{figure}
\subsection{Discussion}\label{subsec:discussion}
The simulation results demonstrate that the proposed approach is effectively structured to handle the multi-phase complexity of aerial manipulation. First, the system achieves precise end-effector point stabilization even when a time-varying trajectory is intentionally applied to the first joint of the robotic arm about the $z$-axis, dynamically proving that the quadrotor's motion and the internal arm adjustments are successfully decoupled. Consequently, this decoupled control framework maintains robust trajectory tracking of the system’s overall $CoM$ while simultaneously regulating the desired remaining arm angles, thereby ensuring that the combined interactive dynamics of the drone and arm do not destabilize the primary flight path. These capabilities ultimately culminate in a highly coordinated, seamless pick-and-place sequence, where the mathematical smoothness of the planned quintic trajectory prevents aggressive control saturation and guarantees steady, predictable transitions during object engagement.
These results validate the suitability of the proposed controller and trajectory tracking framework for aerial manipulation tasks and provide a foundation for experimental implementation.
\section{Conclusion}
In conclusion, this research successfully validates an integrated control and trajectory tracking framework designed for the unique demands of aerial manipulation. By prioritizing high-precision point stabilization alongside $CoM$ trajectory tracking, the system effectively manages the shifting mass distribution and aerodynamic disturbances inherent in robotic arm operations. The simulation outcomes confirm that the quintic trajectory approach minimizes mechanical stress and control effort, ensuring stability during the critical transition from free flight to load-bearing manipulation. Ultimately, this framework provides a robust and scalable foundation for the transition from theoretical modeling to real-world experimental implementation in complex pick-and-place environments.
\subsection*{Acknowledgments}
This work is supported by the Science and Engineering Research Board (SERB) of the Department of Science and Technology (DST), India under the research grant IIT Palakkad Technology IHub Foundation Technology Development Grant IPTIF/TD/IP/005.
\subsection*{Conflicts of Interest}
The authors declare no conflicts of interest.
\appendix 
\section{Appendix}\label{append}
\subsection{Rotation matrix}
Let $s_i$ and $c_i$ denote $\sin(i)$ and $\cos(i)$, respectively. For composite rotations, $c_{yp}$ represents $\cos(\phi_y)\cos(\phi_p)$, and so on.
This is the standard $ZYX$ (Yaw-Pitch-Roll) Euler transformation.
\begin{equation}
   R_0 = \begin{bmatrix}
       c_y c_p & c_y s_p s_r - s_y c_r & c_y s_p c_r + s_y s_r \\ s_y c_p & s_y s_p s_r + c_y c_r & s_y s_p c_r - c_y s_r \\-s_p & c_p s_r & c_p c_r
   \end{bmatrix}
\end{equation}
Since $R_1 = R_0(\Phi) R_z(\theta_1) R_y(\theta_2)$, it is often clearer to define the relative rotation of the RM in ${\mathcal{F}^b_Q}$,  $R^b_{1} = R_z(\theta_1) R_y(\theta_2)$ first, and then provide the multiplied result. Relative Arm Rotation in the $\mathcal{F}^b_Q$:$$R^b_1 = \begin{bmatrix}c_{\theta_1}c_{\theta_2} & -s_{\theta_1} & c_{\theta_1}s_{\theta_2} \\ s_{\theta_1}c_{\theta_2} & c_{\theta_1} & s_{\theta_1}s_{\theta_2} \\ -s_{\theta_2} & 0 & c_{\theta_2}\end{bmatrix}.$$
\subsection{Calculation of \texorpdfstring{$\bar M$}{M-bar}:}\label{bar M}
Here, the rigorous calculation of the mass matrix of the decoupled system \eqref{new dynamics} is demonstrated, where translation and orientation terms are purely separated.
\begin{align*}
   & \mathcal{B}^\top M(q) \mathcal{B} =\begin{bmatrix}
    I_3 & 0  \\ 
    N(\zeta)^\top &   I_5
\end{bmatrix}\begin{bmatrix}
        \bar{M}_{11} &  \bar{M}_{12}\\
        {\bar{M}_{12}}^\top & \bar{M}_{22}
    \end{bmatrix}\begin{bmatrix}
    I_3 & N(\zeta)  \\ 
    0 &   I_5
\end{bmatrix} 
\\
&= \begin{bmatrix}
    m_{qm}I_3 & \bar M_{12} \\
    m_{qm}N^\top(\zeta)+\bar M_{12}^\top & N^\top(\zeta)\bar M_{12} + \bar M_{22}
\end{bmatrix}\\& \quad \hspace{4cm}\;\;\begin{bmatrix}
    I_3 & N(\zeta)  \\ 
    0 &   I_5
\end{bmatrix} \\
&=\begin{bmatrix}
    m_{qm}I_3 & m_{qm}N(\zeta)+\bar M_{12} \\
    m_{qm}N^\top(\zeta)+\bar M_{12}^\top  & 
    \begin{aligned}
        &\bigl(m_{qm}N^\top(\zeta)+\bar M_{12}^\top \bigr)N(\zeta) \\
        &\;\; + N^\top(\zeta)\bar M_{12} + \bar M_{22}
    \end{aligned}
\end{bmatrix}
\end{align*}
Now, substituting $N(\zeta)$ from \eqref{N and M12} we get,
\begin{align*}
    \mathcal{B}^\top \bar M \mathcal{B} &=\begin{bmatrix}
    I_3 & 0  \\ 
    0 &   N^\top(\zeta)\bar M_{12} + \bar M_{22}
\end{bmatrix}.
\end{align*}
\subsection{Calculation of \texorpdfstring{$\bar G$}{G-bar}} \label{bar G}
The Gravitational vector of the decoupled dynamics \eqref{new dynamics} is computed as follows:
\begin{align}
    \bar G = \mathcal{B}^\top G(q) & = \begin{bmatrix}
    I_3 & 0  \\ 
    N(\zeta)^\top &   I_5
\end{bmatrix}
\begin{bmatrix}
m_{qm}g\mathbf{e}_3 \\m_1 g\Bigr[\frac{\partial\langle\eta_1,\mathbf{e}_3\rangle}{\partial \zeta}\Bigl]^\top
\end{bmatrix}, \nonumber \\
 & = \begin{bmatrix}
    m_{qm} g \mathbf{e}_3\\m_{qm}N^\top \mathbf{e}_3 +m_1 g \Bigr[\frac{\partial\langle\eta_1,\mathbf{e}_3\rangle}{\partial \zeta}\Bigl]^\top
\end{bmatrix}, \label{G_bar column}
\end{align}
to calculate the second term of second row of \eqref{G_bar column}, we re-write equation \eqref{manipulator translation} as:
\begin{align*}
\dot \eta_1 &= \dot \eta_0 + \begin{bmatrix}
        - \widehat{R_0\eta^b_1}T & R_0J_{v1}
    \end{bmatrix} \dot \zeta
\end{align*}
Define $J_{\zeta} =\dot \eta_0 + \begin{bmatrix}
        - \widehat{R_0\eta^b_1}T & R_0J_{v1}
    \end{bmatrix} $
; this yields:
\begin{align} \label{eta_1 dot}
\dot \eta_1 &= \dot \eta_0 + J_{\zeta} \dot \zeta
\end{align}
we have $\eta_1 = \eta_1(\eta_0,\zeta)$; then the total partial derivative is as follows:
\begin{align} \label{total partial derivative}
    \dot \eta_1 = \frac{\partial \eta_1}{\partial \eta_0}\dot \eta_0+\frac{\partial \eta_1}{\partial \zeta}\dot \zeta
\end{align}
Now, comparing \eqref{eta_1 dot} and \eqref{total partial derivative}, this renders
\[\frac{\partial \eta_1}{\partial \zeta} = J_\zeta,\] this can be further expressed as:
\begin{align*}
\mathbf{e}_3^\top\frac{\partial \eta_1}{\partial \zeta} &= \mathbf{e}_3^\top J_\zeta\\
    \frac{\partial \langle\eta_1,\mathbf{e}_3\rangle}{\partial \zeta} &= \mathbf{e}_3^\top\begin{bmatrix}
        - \widehat{R_0\eta^b_1}T & R_0J_{v1}
    \end{bmatrix} \\
    m_1\frac{\partial \langle\eta_1,\mathbf{e}_3\rangle}{\partial \zeta} &= \mathbf{e}_3^\top\begin{bmatrix}
        - m_1\widehat{R_0\eta^b_1}T & m_1R_0J_{v1}
    \end{bmatrix}\\ 
    &= \mathbf{e}_3^\top\begin{bmatrix}
        M_{12} & M_{13}
    \end{bmatrix},
\end{align*}
Therefore, from \eqref{eq: block mass matrix}, the above expression can be seen as
\begin{align} \label{m1*partial eta_1}
    m_1\frac{\partial \langle\eta_1,\mathbf{e}_3\rangle}{\partial \zeta} &= \Bigl(\tilde M_{12}^\top \mathbf{e}_3\Bigr)^\top,
\end{align}
Finally, substituting \eqref{N and M12} and \eqref{m1*partial eta_1} in \eqref{G_bar column} gives:
\begin{align*}
    \bar G = \begin{bmatrix}
        m_{qm}g\mathbf{e}_3 \\ \mathbf{0}_{5\times1}
    \end{bmatrix}.
\end{align*}
\subsection{Skew symmetry of \texorpdfstring{$\dot M_S - 2C_S$}{dot M S minus 2 C S}} \label{skew symmetric of dotM_s-2C_s}
Given the internal rotational dynamics of a QRM system as described by \eqref{eq:shape}:
\[
M_S(\zeta){\ddot \zeta} + C_S(\zeta,\dot{\zeta})\dot{\zeta} = \tau_S
\]
where $M_S(\zeta)$ is the symmetric, positive-definite inertia matrix, and the matrix $\dot{M}_S - 2C_S$ is skew-symmetric if $C_S$ is defined using Christoffel symbols of the first kind.
From the definition of the Coriolis Matrix, the elements $C_{ij}$ of the Coriolis matrix $C_S$ are defined by the configuration vector $\zeta$ and the velocity $\dot{\zeta}$:
\[
C_{ij} = \sum_{k=1}^{5} \Gamma_{ijk} \dot{\zeta}_k
\]
where the Christoffel symbols of the first kind, $\Gamma_{ijk}$, are derived from the partial derivatives of the inertia matrix elements:
\[
\Gamma_{ijk} = \frac{1}{2} \left( \frac{\partial M_{S_{ij}}}{\partial \zeta_k} + \frac{\partial M_{S_{ik}}}{\partial \zeta_j} - \frac{\partial M_{S_{kj}}}{\partial \zeta_i} \right)
\]
Now applying the chain rule, the time derivative of the $(i,j)$-th element of the inertia matrix $M_S$ is expressed as:
\[
\dot{M}_{S_{ij}} = \sum_{k=1}^{5} \frac{\partial M_{S_{ij}}}{\partial \zeta_k} \dot{\zeta}_k
\]
Let $n = \dot{M}_S - 2C_S$, thus
$n_{ij}$ denote the $(i,j)$-th element of the matrix $n$. Substituting the expressions from the previous steps:
\begin{equation*}
\begin{aligned}
    n_{ij} =& \dot{M}_{S_{ij}} - 2C_{ij},\\
    =&\sum_k \frac{\partial M_{S_{ij}}}{\partial \zeta_k} \dot{\zeta}_k - 2 \sum_k \frac{1}{2} \left( \frac{\partial M_{S_{ij}}}{\partial \zeta_k} + \frac{\partial M_{S_{ik}}}{\partial \zeta_j} - \frac{\partial M_{S_{kj}}}{\partial \zeta_i} \right) \dot{\zeta}_k
    ,\\
   =& \sum_k \left[ \frac{\partial M_{S_{ij}}}{\partial \zeta_k} - \left( \frac{\partial M_{S_{ij}}}{\partial \zeta_k} + \frac{\partial M_{S_{ik}}}{\partial \zeta_j} - \frac{\partial M_{S_{kj}}}{\partial \zeta_i} \right) \right] \dot{\zeta}_k\\
   =& \sum_k \left( \frac{\partial M_{S_{kj}}}{\partial \zeta_i} - \frac{\partial M_{S_{ik}}}{\partial \zeta_j} \right) \dot{\zeta}_k
\end{aligned}
\end{equation*}
using the symmetry property of the inertia matrix ($M_{S_{ki}} = M_{S_{ik}}$ and $M_{S_{jk}} = M_{S_{kj}}$):
\[
n_{ji} = \sum_k \left( \frac{\partial M_{S_{ik}}}{\partial \zeta_j} - \frac{\partial M_{S_{kj}}}{\partial \zeta_i} \right) \dot{\zeta}_k
\]
henceforth, it is evident that:
\[n_{ij} = -n_{ji}\]
Since the transpose of the matrix equals its negative ($n^\top = -n$), the matrix $\dot{M}_S - 2C_S$ is skew-symmetric. 
\subsection{Mathematical interpretation of thrust}
\begin{corollary} \label{tauL = uct}
Given the attitude control \eqref{eq:ps_shape_ctrl}, which ensures $R_0 \to R^d_{0}$, where $R^d_{0}$ is the computed rotation matrix, and from the construction of $R^d_{0}$, we have
\begin{equation}
R^d_{0} e_3 = \frac{u_{ct}}{\|u_{ct}\|},
\end{equation}
then it follows that $\tau_L = u_{ct}$.
\end{corollary}
%
\begin{proof}
The stabilizing thrust force for locked subsystem dynamics \eqref{eq:locked_ss}, is used to define the reference of the vehicle's $z$-axis (denoted by $\mathbf{r}_{3}^{d}$) as
\begin{equation}
    \mathbf{r}_{3}^{d} = \frac{u_{ct}}{\|u_{ct}\|}.
\end{equation}
Note that $\mathbf{r}_{3}^{d}$ is nothing but the third column of $R^d_{0}$. Under control moments from \eqref{eq:ps_shape_ctrl}, which ensures $R_0 \to R^d_{0}$. Hence, we can write
\begin{equation}
\tau_{L} = \langle u_{ct}, R^d_{0} e_3 \rangle R^d_{0} e_3.
\end{equation}
This can be rewritten as
\begin{equation}
\tau_{L} = (R^d_{0} e_3)^\top u_{ct} \, (R^d_{0} e_3).
\end{equation}
Now substitute $R^d_{0} e_3 = \frac{u_{ct}}{\|u_{ct}\|}$ into the above expression:
\begin{equation}
\tau_L = \left(\frac{u_{ct}}{\|u_{ct}\|}\right)^\top u_{ct} \left(\frac{u_{ct}}{\|u_{ct}\|}\right).
\end{equation}
This gives
\begin{equation}
\tau_L = \frac{(u_{ct}^\top u_{ct})\, u_{ct}}{\|u_{ct}\|^2}.
\end{equation}
Since $u_{ct}^\top u_{ct} = \|u_{ct}\|^2$, we finally obtain $\tau_L = u_{ct}$.
\end{proof}
\subsection{Boundedness of matrix N}
\begin{claim} \label{Norm of N}
    For N in \eqref{N and M12}, such that $0\leq \Vert N\rVert^2 \leq \frac{5m_1^2l^2}{4m_{qm}^2}$.
\end{claim}
\begin{proof}
    From an identity of $L_2$ norm,
    \begin{equation} \label{norm N2}
        \|N\|^2 = {\lambda_{\max}(N^T N)} \leq trace(N^\top N)
    \end{equation}
 here, 
 $N^\top N = \frac{1}{m_{qm}^2}\begin{bmatrix}
M_{12}^\top M_{12} & M_{12}^\top M_{13}\\
M_{13}^\top M_{12} & M_{13}^\top M_{13}
 \end{bmatrix}, $
 This gives
 \begin{align} \label{trace NT N}
     trace(N^\top N ) = \frac{1}{m_{qm}^2}[trace(M^\top_{12}M_{12}) +trace(M^\top_{13}M_{13})]
 \end{align}
 calculating terms of \eqref{trace NT N},
 \begin{align*}
     M_{12}^\top M_{12} &= m_1^2 T^\top \widehat{R_0\eta_{1}^{b}}^\top \widehat{R_0\eta_{1}^{b}} T \\
     &= - \frac{m_1^2l^2}{4}T^\top(\widehat{R_0R_{1}^{b}e_3})^2 T\\
     trace(M_{12}^\top M_{12}) &= - \frac{m_1^2l^2}{4} trace(T^\top(\widehat{R_0R_{1}^{b}e_3})^2 T)
 \end{align*}
Note that $-3 \leq trace(T^\top\widehat{({R_0R_{1}^{b}e_3}})^2 T) \leq 3 $; this implies 
 \[- \frac{3}{4}m_1^2l^2\leq trace(M_{12}^\top M_{12}) \leq \frac{3}{4}m_1^2l^2\]
Similarly, the second term will be bounded as 
 \[-\frac{m_1^2l^2}{2}\leq trace(M_{13}^\top M_{13}) \leq \frac{m_1^2l^2}{2} \]
 Therefore, the boundedness on \eqref{trace NT N} takes the form
 \[-\frac{5m_1^2l^2}{4m_{qm}^2} \leq trace(N^\top N) \leq \frac{5m_1^2l^2}{4m_{qm}^2}\]
 Finally, from \eqref{norm N2} and $\Vert N\rVert^2 \geq 0$, this renders
 \[0\leq \Vert N\rVert^2 \leq trace(N^\top N) \leq \frac{5m_1^2l^2}{4m_{qm}^2}.\]
\end{proof}
\bibliographystyle{IEEEtran}
\bibliography{IJc}
\end{document}